\documentclass[11pt]{article}

\usepackage{fullpage,amsmath,xspace,amssymb,
  algorithm,algorithmicx,comment,graphicx,caption}
\usepackage[noend]{algpseudocode}

\usepackage{umoline}
\usepackage[usenames,svgnames]{xcolor}

\usepackage{color,transparent}

\usepackage{hyperref}
\hypersetup{
     colorlinks=true,      		% false: boxed links; true: colored links
     linkcolor=Salmon,        	% color of internal links
     citecolor=blue,            % color of links to bibliography
     filecolor=blue,      		% color of file links
     urlcolor=cyan,           	% color of external links
}

\newtheorem{Thm}{Theorem}
\newtheorem{Lem}[Thm]{Lemma}
\newtheorem{Cor}[Thm]{Corollary}
\newtheorem{Prop}[Thm]{Proposition}
\newtheorem{Fact}[Thm]{Fact}

\newtheorem{Def}{Definition}
\newenvironment{proof}{\noindent {\textbf{Proof }}}{$\Box$ \medskip}

\newcommand\mbC{\mbox{$\mathbb{C}$}}

\newcommand\Qb{\mathcal{H}}

\newcommand\defeq{=}

\newcommand\rank{\mbox{\tt {rank}}\xspace}

\newcommand\ket[1]{| #1 \rangle}
\newcommand\bra[1]{\langle #1 |}
\newcommand\qip[2]{\langle #1 | #2 \rangle}

\newcommand\Id{\mathbb{I}}
\newcommand\Ref[1]{Ref.~\cite{#1}}

\newcommand{\Empty}{\circleddash}
\newcommand{\Bad}{X}

\newcommand\QMA{{\mathrm{QMA}}}

\newcommand{\suppress}[1]{}

\newcommand{\Eq}[1]{Eq.~(\ref{#1})}
\newcommand{\gs}{\Gamma}

\newcommand\slide{{\mathsf{slide}}}
\newcommand\extt{{\mathsf{ext}}}
\newcommand\spa{{\mathsf{span}}}
\newcommand\pro{\mathsf{prop}}

\newcommand\supp{\mathsf{supp}}

\begin{document}
\title{\textbf{Linear time algorithm for quantum 2SAT}}
\author{
Itai Arad\thanks{Centre for Quantum Technologies, National University of Singapore, 
Singapore 117543 ({\tt arad.itai@fastmail.com}).}
\and
Miklos Santha\thanks{CNRS, IRIF, Universit\'e Paris Diderot 75205 Paris, France;  and 
Centre for Quantum Technologies, National University of Singapore, 
Singapore 117543 ({\tt miklos.santha@gmail.com}).}
\and
Aarthi Sundaram\thanks{Centre for Quantum Technologies, National University of Singapore, 
Singapore 117543 ({\tt aarthims@gmail.com}).}
\and
Shengyu Zhang\thanks{Department of Computer Science and Engineering, The Chinese University of Hong Kong, Shatin, N.T., Hong Kong. Email: {\tt syzhang@cse.cuhk.edu.hk}}
}
%\date{}
\maketitle

\begin{abstract}
  A canonical result about satisfiability theory is that the 2-SAT
  problem can be solved in linear time, despite the NP-hardness of
  the 3-SAT problem. In the quantum 2-SAT problem, we are given a
  family of 2-qubit projectors $\Pi_{ij}$ on a system of $n$ qubits,
  and the task is to decide whether the Hamiltonian $H=\sum
  \Pi_{ij}$ has a 0-eigenvalue, or it is larger than $1/n^\alpha$
  for some $\alpha=O(1)$. The problem is not only a natural
  extension of the classical 2-SAT problem to the quantum case, but
  is also equivalent to the problem of finding the ground state of
  2-local frustration-free Hamiltonians of spin $\frac{1}{2}$, a
  well-studied model believed to capture certain key properties in
  modern condensed matter physics. While Bravyi has shown that the
  quantum 2-SAT problem has a classical polynomial-time algorithm,
  the running time of his algorithm is $O(n^4)$. In this paper we
  give a classical algorithm with linear running time in the number
  of local projectors, therefore achieving the best possible
  complexity.
\end{abstract}

%%%%% to set the title page uncounted
\thispagestyle{empty}
\clearpage
\setcounter{page}{1}

%XXXXXXXXXXXXXXXXXXXXXXXXXXXXX intro_short XXXXXXXXXXXXXXXXXXXXXXXXXX

%%%%%%%%%%%%%%%%%%%%%%%%%%%%%%%%%%%%%%%%%%%%%%%%%%%%%%%%%%%%%%%%%%%%%%%%
\section{Introduction}
\label{sec:intro}

Various formulations of the satisfiability problem of Boolean
formulae arguably constitute the center piece of classical
complexity theory. In particular, a great amount of attention has
been paid to the SAT problem, in which we are given a formula in the form of 
%a formula in conjunctive normal form (CNF) \footnote{A formula in CNF is a
a conjunction of \emph{clauses}, where each clause is a disjunction of
\emph{literals} (variables or negated variables), and the task is
to find a satisfying assignment if there is one, or prove that none
exists when the formula is unsatisfiable. In the case of the $k$-SAT
problem, where $k$ is a positive integer, in each clause the number
of literals is at most $k$. While $k$-SAT is an NP-complete 
problem~\cite{Cook71,Karp72,Lev73} when $k\geq 3$, the problem 2-SAT
is well-known to be efficiently solvable. 

Polynomial time algorithms for 2-SAT come in various flavors. Let us
suppose that the input formula has $n$ variables and $m$ clauses.
The algorithm of Krom~\cite{Krom76} based on the resolution
principle and on transitive closure computation decides if the
formula is satisfiable in time $O(n^3)$ and finds a satisfying
assignment in time $O(n^4)$. The limited backtracking technique of
Even, Itai and Shamir~\cite{EIS76} has linear time complexity in
$m$, as well as the elegant procedure of Aspvall, Plass and
Tarjan~\cite{APT79} based on computing strongly connected components
in a graph. A particularly simple randomized procedure of complexity
$O(n^2)$ is described by Papadimitriou~\cite{Pap91}.

For our purposes the Davis-Putnam procedure~\cite{DP60} is of
singular importance. This is a resolution-principle based general
SAT solving algorithm, which with its refinement due to Davis,
Putnam, Logemann and Loveland~\cite{DPLL62}, forms even today the
basis for the most efficient SAT solvers. While on general SAT
instances it works in exponential time, on 2-SAT formulae it is of
polynomial complexity.

The high level description of the procedure for 2-SAT is relatively
simple.  Let us suppose that our formula $\phi$ contains only
clauses with two literals. Pick an arbitrary unassigned variable
$x_i$ and assign $x_i = 0$. The formula is simplified: a clause
$(\bar{x}_i \vee x_j)$ becomes true and therefore can be removed,
and a clause $(x_i \vee x_j)$ forces $x_j = 1$. This can be, in
turn, propagated to other clauses to further simplify the formula
until a contradiction is found or no more propagation is possible.
If no contradiction is found and the propagation stops with the
simplified formula $\phi_0$, then we recurse on the satisfiabilty of
$\phi_0$.  Otherwise, when a contradiction is found, that is
at some point the propagation assigns two different values to the
same variable, we reverse the choice made for $x_i$, and
propagate the new choice $x_i=1$. If this also leads to
contradiction we declare $\phi$ unsatisfiable, otherwise we recurse
on the result of this propagation, the simplified formula $\phi_1$.

There is a deep and profound link between $k$-SAT formulas
and
$k$-local Hamiltonians, the central objects of condensed matter
physics. A $k$-local Hamiltonian on $n$ qubits is a Hermitian
operator of the form $H=\sum_{i=1}^m
h_i$, where each $h_i$ is by itself a
Hermitian operator acting non-trivially on at most $k$ qubits. Local
Hamiltonians model the local interactions between quantum spins. Of
central importance is the minimal eigenstate of the Hamiltonian,
known as the \emph{ground state}, and its associated eigenvalue,
known as the \emph{ground energy}. The ground state governs much of
the low temperature physics of the system, such as quantum phase
transitions and collective quantum
phenomena~\cite{ref:Sachdev,ref:Vidal-2003}. Finding the ground
state of a local Hamiltonian shares important similarities with the
$k$-SAT problem: in both problems we are trying to find a global
minimum of a set of local constraints.  This
connection with complexity theory is of physical significance.
Indeed, with
the advent of quantum information theory and quantum complexity
theories, it has become clear that the complexity of finding the
ground state and its energy is intimately related to its
entanglement structure. In recent years, much attention has been
devoted into understanding this structure, revealing a rich an
intricate behaviour such as area laws~\cite{ECP10} and topological
order~\cite{ref:toric}. 

The connection between classical $k$-SAT and quantum local
Hamiltonian was formalized by Kitaev~\cite{Kit02} who introduced the
$k$-local Hamiltonian problem: one is given a $k$-local Hamiltonian
$H$, along with two constants $a < b$ such that $b-a>1/n^\alpha$ for
some constant $\alpha$. It is promised that the ground energy of $H$
is at most $a$ (the YES case) or is at least $b$ (the NO case), and
the task is to decide which case holds.  Broadly speaking, given a
quantum state $\ket{\psi}$, the energy of a local term
$\bra{\psi}h_i\ket{\psi}$ is a measure of how much $\ket{\psi}$
``violates'' $h_i$, hence the ground energy is the quantum analog of
the minimal number of violations in a classical $k$-SAT. Therefore,
in spirit, the $k$-local Hamiltonian problem corresponds to
MAX-$k$-SAT, and indeed Kitaev has shown~\cite{Kit02} that 5-local
Hamiltonian is QMA-complete, where the complexity class QMA is the
quantum analogue of classical class MA, the probabilistic version of
NP.

The problem \emph{quantum $k$}-SAT, the quantum analogue of $k$-SAT,
is a close relative of the $k$-local Hamiltonian problem. Here we
are given a $k$-local Hamiltonian that is made of $k$-local
\emph{projectors}, $H=\sum_{i=1}^m Q_i$, and we are asked whether
the ground energy is 0 or it is larger than $b=1/n^\alpha$ for some
constant $\alpha$. Notice that in the YES case, the energy of all
projectors at the ground state is necessarily 0, since by
definition, projectors are non-negative operators. Classically, this
corresponds to a perfectly satisfiable formula. Physically, this is
an example of a \emph{frustration-free} Hamiltonian, in which the
global ground state is also a ground state of every local term. 
Bravyi~\cite{Bra06} has shown that quantum $k$-SAT was
$\QMA_1$-complete for $k \geq 4$, where $\QMA_1$ stands for $\QMA$
with one-sided error (that is on YES instances the verifier accepts
with probability 1). The $\QMA_1$-completeness of quantum 3-SAT was
recently proven by Nagaj~\cite{GN13}. 

This paper is concerned with the quantum 2-SAT problem, which we
will also denote simply by Q2SAT. One major result concerning this
problem is due to Bravyi~\cite{Bra06}, who has proven that it
belongs to the complexity class P. More precisely, he has proven
that Q2SAT can be decided by a deterministic algorithm in time
$O(n^4)$, together with a ground state that has a polynomial
classical description. In the case of Q2SAT, the
Hamiltonian is given as a sum of 2-qubits projectors; each projector
is defined on a 4-dimensional Hilbert space and can therefore be of
rank 1, 2 or 3. In this paper, we give an algorithm for Q2SAT of \textit{linear} complexity.

\begin{Thm}
\label{thm:intro} 
  There is a deterministic algorithm for ${\rm Q2SAT}$ whose running
  time is $O(n+m)$ where $n$ is the number of variables and m is the number of local terms in the
  Hamiltonian.
\end{Thm}

Our algorithm shares the same trial and error approach of the Davis-Putnam procedure for classical 2SAT, but handles many difficulties arising in the quantum setting. First, a ground state of Q2SAT input may be entangled, some distinctive feature that classical 2SAT does not have. Thus the idea of setting some qubit to certain state and propagating from there does not have foundation at the first place. Indeed, if a rank-3 projection leaves the only allowed state entangled, then any ground state is entangled in those two qubits. We overcome this by showing a \textit{product-state theorem}, which asserts that for any frustration-free Q2SAT instance $H$ that contains only rank-1 and rank-2 projectors, there always exists a ground state in the form of a tensor product of \textit{single-qubit} states. 

This structural theorem grants us the following approach: We try some candidate solution $\ket{\psi}_i$ on a qubit $i$, and propagate this along the graph. If no contradiction is found, it turns out that we can detach the explored part and recurse on the rest of the graph. If a contradiction is found, then we can identify two candidates $(i,\ket{\psi}_i)$ and $(j,\ket{\phi}_j)$ such that either assigning $\ket{\psi}_i$ to qubit $i$ or assigning $\ket{\phi}_j$ to qubit $j$ is correct, if there exists a solution at all. More details follow next.

To illustrate the main idea of our algorithm, let us suppose that
the input contains only projectors of rank at most two. Such a system
can be further simplified to a system consisting only of rank-1
projectors, by writing every rank-2 projector as a sum of two rank-1
projectors. Consider, for example, qubits 1 and 2 and a rank-1
projector $\Pi_{12} = \ket{\psi}\bra{\psi}$ over these qubits. The
product-state theorem implies that it suffices to search for a product
ground state. Thus on the first two qubits, we are looking for
states $\ket{\alpha},\ket{\beta}$ such that
$(\bra{\alpha}\otimes\bra{\beta})\Pi_{12}
(\ket{\alpha}\otimes\ket{\beta})=0$, which is equivalent to
$\bra{\alpha}\otimes\bra{\beta}\cdot\ket{\psi} = 0$. In other words,
we look for a product state $\ket{\alpha}\otimes\ket{\beta}$ that is
perpendicular to $\ket{\psi}$. Assume that we have assigned qubit 1
with the state $\ket{\alpha}$ and we are looking for a state
$\ket{\beta}$ for qubit 2. The crucial point, which enables us to
solve Q2SAT efficiently, is that just like in the classical case,
there are only two possibilities: (i) for any $\ket{\beta}$, the
state $\ket{\alpha}\otimes\ket{\beta}$ is perpendicular to
$\ket{\psi}$, or (ii) there is only one state $\ket{\beta}$ (up to
an overall complex phase), for which
$(\bra{\alpha}\otimes\bra{\beta})\cdot\ket{\psi}=0$. The first case
happens if and only if $\ket{\psi}$ is by itself a product state of
the form $\ket{\psi}=\ket{\alpha^\perp}\otimes\ket{\xi}$, where
$\ket{\alpha^\perp}$ is perpendicular to $\ket{\alpha}$ and
$\ket{\xi}$ is arbitrary. If the second case happens, we say that
state $\ket{\alpha}$ is propagated to state $\ket{\beta}$ by the
constraint state $\ket{\psi}$.

The above dichotomy enables us to propagate a product state
$\ket{s}$ on part of the system until we either reach a
contradiction, or find that no further propagation is possible and
we are left with a smaller Hamiltonian $H_s$. This smaller
Hamiltonian consists of a subset of the original projectors,
\emph{without introducing new projectors}. It turns out that once an
edge is checked for potential propagation, then no matter whether a
propagation happens along the edge, the edge can be safely removed
without changing the satisfiability. Thus the satisfiability of the
original Hamiltonian $H$ is the same as that of the smaller
Hamiltonian $H_s$.

We still need to specify how the state $\ket{\alpha}$ is chosen to
initialize the propagation. An idea is to begin with projectors
$\ket{\psi}\bra{\psi}$ for which $\ket{\psi}$ is a product state
$\ket{\alpha}\otimes\ket{\beta}$. In such cases a product state
solution must either have $\ket{\alpha^\perp}$ at the first qubit or
$\ket{\beta^\perp}$ at the second. To maintain a linear running
time, we propagate these two choices simultaneously until one of the
propagations stops without contradiction, in which case the
corresponding qubit assignment is made final. If both propagations
end with contradiction, the input is rejected.

The more interesting case of the algorithm happens when we have only
entangled rank-1 projectors. What should our initial state be then?
We make an arbitrary assignment (say, $\ket{0}$) to any of the still
unassigned qubits and propagate this choice. If the propagation ends
without contradiction, we recurse. If a contradiction is found then
we confront a challenging problem. In the classical case we could
reverse our choice, say $x_0 = 0$, and try the other possibility,
$x_i = 1$. But in the quantum case we have an infinite number of
potential assignment choices.  The solution is found by the
following observation: Whenever a contradiction is reached, it can
be attributed to a cycle of entangled projectors in which the
assignment has propagated from qubit $i$ along the cycle and
returned to it with another value. Then using the techniques of
`sliding', which was introduced in Ref.~\cite{JWZ11}, one can show
that this cycle is equivalent to a system of one double edge and a
`tail' (see Fig.~\ref{fig:contradicting-cycle}). Using a simple
structure lemma, we are guaranteed that at least one of the
projectors of the double edge can be turned into a product state
projector, which, as in the previous stage, gives us two possible
free choices.

\begin{figure}
  \begin{center}
    \includegraphics[scale=1.0]{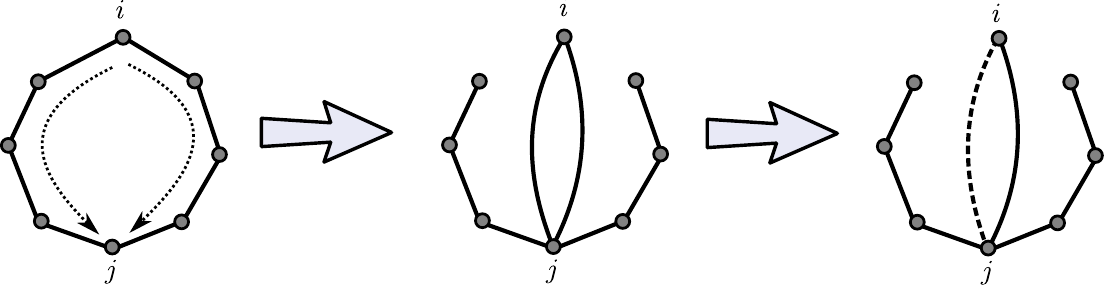}  
  \end{center}
  \caption{\label{fig:contradicting-cycle} Handling a contradicting
  cycle: we  slide edges that touch $i$ along the two paths to $j$
  until we get
  a double edge with a `tail'. We then use a structure lemma to
  deduce that at least one of these edges can be written as a
  product projector.}
\end{figure}

%The only thing remaining is to take care of the rank-3 projectors. This happens once and for all at the very beginning of the algorithm. The procedure itself is rather simple since there is only one way to satisfy a rank-3 projector. Once all the necessary assignments are made, we propagate. If a contradiction is found the input is rejected, otherwise we are left with projectors of rank at most two.

Let us state here that our algorithm works in the algebraic model of
computation: we suppose that every arithmetic operation on complex
numbers can be done in unit time. There are several ways to work in
a more realistic model.  Bravyi~\cite{Bra06} suggests considering
bounded degree algebraic numbers, in which case the length of the
representations and the cost of the operations should be analyzed.
Another possibility would be to consider complex numbers with
bounded precision in which case exact computation is no more
possible and therefore an error analysis should be made. By choosing
a more convenient computational model, we disregard these issues.

Classically, Davis-Putnam \cite{DP60} and DPLL algorithms
\cite{DPLL62} are widely-used heuristics, forming the basis of
today's most efficient solvers for general SAT. For quantum $k$-SAT,
it could also be a good heuristic if we try to find product-state
solutions, and in that respect our algorithm makes the first-step
exploration. 

Simultaneously and independently from our work and approximately at the same time, de Beaudrap and
Gharibian~\cite{BG} have also presented a linear time algorithm for quantum 2SAT.
The main difference between the two algorithms is how they deal with instances
with only entangled rank-1 projectors. Contrarily to us,~\cite{BG} handles these instances
by using transfer matrix techniques to find discretizing cycles~\cite{ref:Laumann09}.

%XXXXXXXXXXXXXXXXXXXXXXXXXXXXX prelim_short XXXXXXXXXXXXXXXXXXXXXXXXXX

% token holder: no one

%%%%%%%%%%%%%%%%%%%%%%%%%%%%%%%%%%%%%%%%%%%%%%%%%%%%%%%%%%%%%%%%%%%%%
\section{Preliminaries}

%====================================================================
\subsection{Notation}

We will use the notation $[n]= \{1, \ldots,n\}$.  For a graph $G=
(V,E)$, and for a subset $U \subseteq V$ of the vertices, we denote
by $G(U)$ the subgraph induced by $U$. Our Hilbert space is defined
over $n$ qubits, and is written as
$\mathcal{H}=\Qb_1\otimes\Qb_2\otimes\cdots\otimes\Qb_n$, where
$\Qb_i$ is the two-dimensional Hilbert space of the $i^{th}$ qubit.
We shall often write $\ket{\alpha}_i$ to emphasize that the 1-qubit
state $\ket{\alpha}$ lives in $\Qb_i$. Similarly, $\ket{\psi}_{ij}$
denotes a 2-qubit state that lives in $\Qb_i\otimes\Qb_j$. For a
1-qubit state $\ket{\alpha} = \alpha_0\ket{0} + \alpha_1\ket{1}$, we
define its \emph{perpendicular state} as
$\ket{\alpha^\bot}\defeq\alpha_1\ket{0} - \alpha_0\ket{1}$. 
%\mnote{Seems to me incorrect: what if for example $\alpha_0 = i$ and $\alpha_1 = 1$?
%Then $\bra{\alpha} \ket{\alpha^\bot} = 2i$. I would put back the original 
%$\ket{\alpha^\bot}\defeq\alpha_1\ket{0} - \alpha_0\ket{1}$.}

We shall denote local projectors either by $\Pi_{ij}$, or by
$\Pi_e$, where $e=(i,j)$. When $i<j$, $\Pi_{ij}$ is a 2-local
projector on the qubits $i,j$; it can be written as $\Pi_{ij} =
\hat{\Pi}_{ij}\otimes \Id_{rest}$, where $\hat{\Pi}_{ij}$ is a
projector working on $\Qb_i\otimes \Qb_j$ and $\Id_{rest}$ is the
identity operator on the rest of the system. Similarly, when $i=j$,
$\Pi_{ii}=\hat{\Pi}_{ii}\otimes\Id_{rest}$, where $\hat{\Pi}_{ii}$
is a projector defined in $\Qb_i$.  Often, in order not to overload
the notation, we shall use $\Pi_{ij}$ instead of $\hat{\Pi}_{ij}$,
even when acting on states in $\Qb_i\otimes\Qb_j$. Similarly, with a
slight abuse the notation, we define the \emph{rank} of a projector
$\Pi_e$ {to be} the dimension of the subspace that its local
projector $\hat{\Pi}_e$ projects to, and it will be denoted by
$\rank(\Pi_e)$. We call a rank-1 projector $\Pi_e =
\ket{\psi}\bra{\psi}$, \textit{entangled}
if $\ket{\psi}$ is an entangled state, and \textit{product} if
$\ket{\psi}$ is a product state.

%====================================================================
\subsection{The Q2SAT problem}

A quantum 2-SAT Hamiltonian on an $n$-qubit system is a Hermitian
operator $H = \sum_{e \in I} \Pi_e$, for some $I \subseteq \{(i,j)
\in [n] \times [n] :1 \leq i \leq j \leq n\}$.  We suppose that
$\rank(\Pi_{ii}) =1$, for all $(i,i) \in I$, and $0 <
\rank(\Pi_{ij}) < 4$, for all $(i,j) \in I$ when $i<j$. The
single-qubit projectors of $H$ as well as its 2-qubit projectors of
rank-3 are called \emph{maximal} rank.

The \textit{ground energy} of a Hamiltonian $H = \sum_{e \in I}
\Pi_{e}$ is its smallest eigenvalue, and a \textit{ground state} of
$H$ is an eigenvector corresponding to the smallest eigenvalue. The
subspace of the ground states is called the \textit{ground space}. A
Hamiltonian is \textit{frustration-free} if it has a ground state
that is also simultaneously the ground state of all local terms. As
explained in the introduction, if the Hamiltonian is made of local
projectors, it is frustration-free if and only if there is a state
that is a mutual zero eigenstate of all projectors, which happens if
and only if the ground energy is 0. Therefore, if $\ket{\gs}$ is a
ground state of a frustration-free quantum 2-SAT Hamiltonian,
$\Pi_{e}\ket{\gs} = 0$ for all $e \in I$. 
We can also view each local projector as a {\em constraint} on at most two qubits, 
then a ground state satisfies every constraint.
%\mnote{I think that mentioning that we view the projectors as constraints, as it is done in short10,
%is not superfluous: we have for example the notion of the constraint graph. 
%We don't have to say it in italic, as it were a real definition, just in standard roman.
%As for the  {\em satisfaction} of projectors by a quantum state, I agree that we can eliminate here,
%but we should be coherent with the resolution of the problem raised in Section 2.4}

It turns out that for the representation of the 2QSAT Hamiltonian, 
it will be helpful to eliminate the rank-2 projectors by 
decomposing each one of them into a sum of two rank-1 projectors.
For every $(i,j) \in I$ such that $\rank(\Pi_{ij}) = 2$, let
$\Pi_{ij} = \Pi_{ij,1} + \Pi_{ij,2}$, where $\Pi_{ij,1}$ and $
\Pi_{ij,2}$ are rank-1 projectors. Such projectors can be found in
constant time. We therefore suppose without loss of generality that
$H$ is specified by
\begin{align*}
  H = \sum_{ \rank(\Pi_{ij}) \neq 2} \Pi_{ij}
    + \sum_{ \rank(\Pi_{ij}) = 2} (\Pi_{ij,1} + \Pi_{ij,2}) \,,
\end{align*}
which we call the \emph{rank-$1$ decomposition} of $H$.

To the rank-1 decomposition we associate a weighted, directed
multigraph with self-loops $G(H) = (V,E,w)$, the \emph{constraint
graph} of $H$. By definition $V =\{i \in [n] : \exists j \in [n]
{\rm ~ such ~ that ~}(i,j) \in I {\rm ~or~} (j,i) \in I\}$, For
every rank-3 and rank-1 projector acting on two qubits, there is an
edge in each direction between the two nodes representing
them. For every projector acting on a single qubit, there is a
self-loop. Finally, for every rank-2 projector, there are two
parallel edges in each direction between nodes representing its
qubits. Because of the parallel edges, $E$ is not a subset of $V
\times V$. Formally, $E = E_1 \cup E_2$ where
\begin{align*}
  E_1 =  \{ (i,j) \in [n] \times [n] : (i,j) \in I 
    {\rm ~and~} \rank(\Pi_{ij}) \in \{1,3\},  
    \mbox{{\rm ~or~}} (j,i) \in I {\rm ~and~}
  \rank(\Pi_{ji}) \in \{1,3\} \},
\end{align*}
and
\begin{align*}
  E_2 = \{ (i,j,b) \in [n] \times [n] \times [2] : 
    (i,j) \in I {\rm ~and~} \rank(\Pi_{ij}) =2, 
    \mbox{{\rm ~or~}}  (j,i) \in I 
    {\rm ~and~} \rank(\Pi_{ji}) =2 \}.
\end{align*}
We say that an edge $e \in E$ \emph{goes from} $i$ \emph{to} $j$ if
$e \in \{ (i,j), (i,j,1), (i,j,2) \}$. For a projector $\Pi$ acting
on two qubits, we define its \emph{reverse} projector {$\Pi^{rev}$}
by $\Pi^{rev}\ket{\alpha}\ket{\beta} = \Pi\ket{\beta}\ket{\alpha}$,
and for $i \leq j$ and $b \in [2]$, we set $\Pi_{ji}=\Pi_{ij}^{rev}$
and $\Pi_{jib}=\Pi_{ijb}^{rev}$. Then for an edge $(i,j)$, its
\emph{weight} is defined as $w(i,j) = \Pi_{ij}$, and analogously for
an edge $(i,j,b)$, we set $w(i,j,b) = \Pi_{ijb}$. 
%\mnote{In the figure the list of verices should also be a linked list}
\begin{center}
\begin{figure}[htb]
	\begin{minipage}[b]{0.35\textwidth}
    \def\svgwidth{180pt}
  	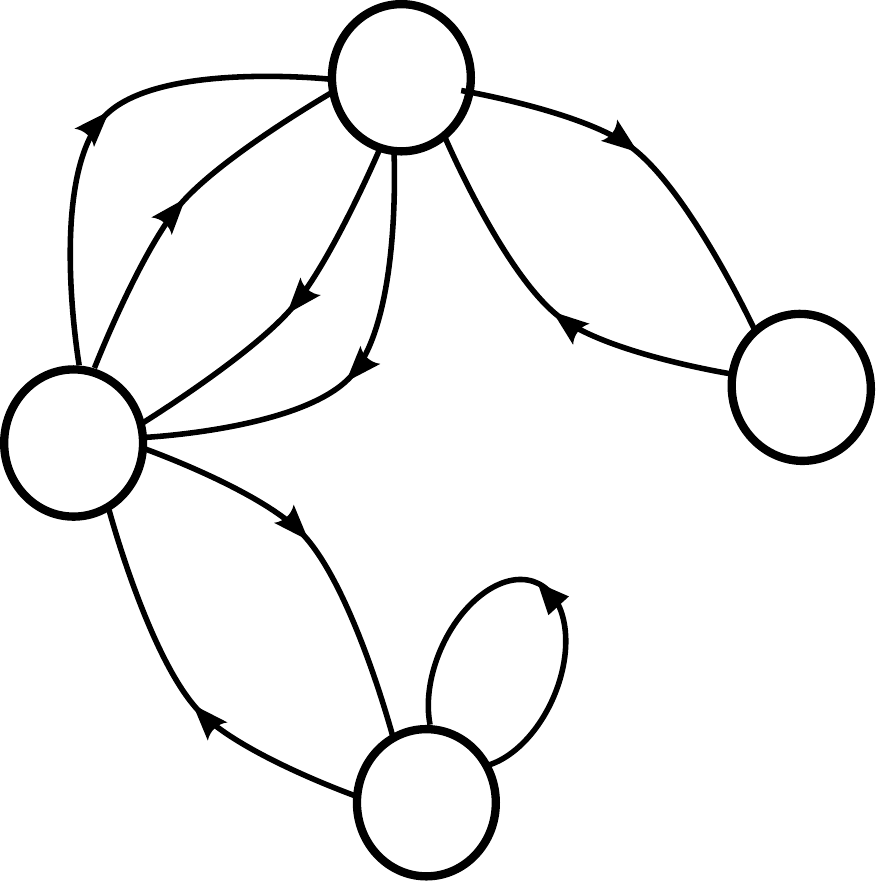
  	\caption*{(a)}
 	\end{minipage}
 	\begin{minipage}[b]{0.65\textwidth}
  	\centering
  	\def\svgwidth{250pt}
  	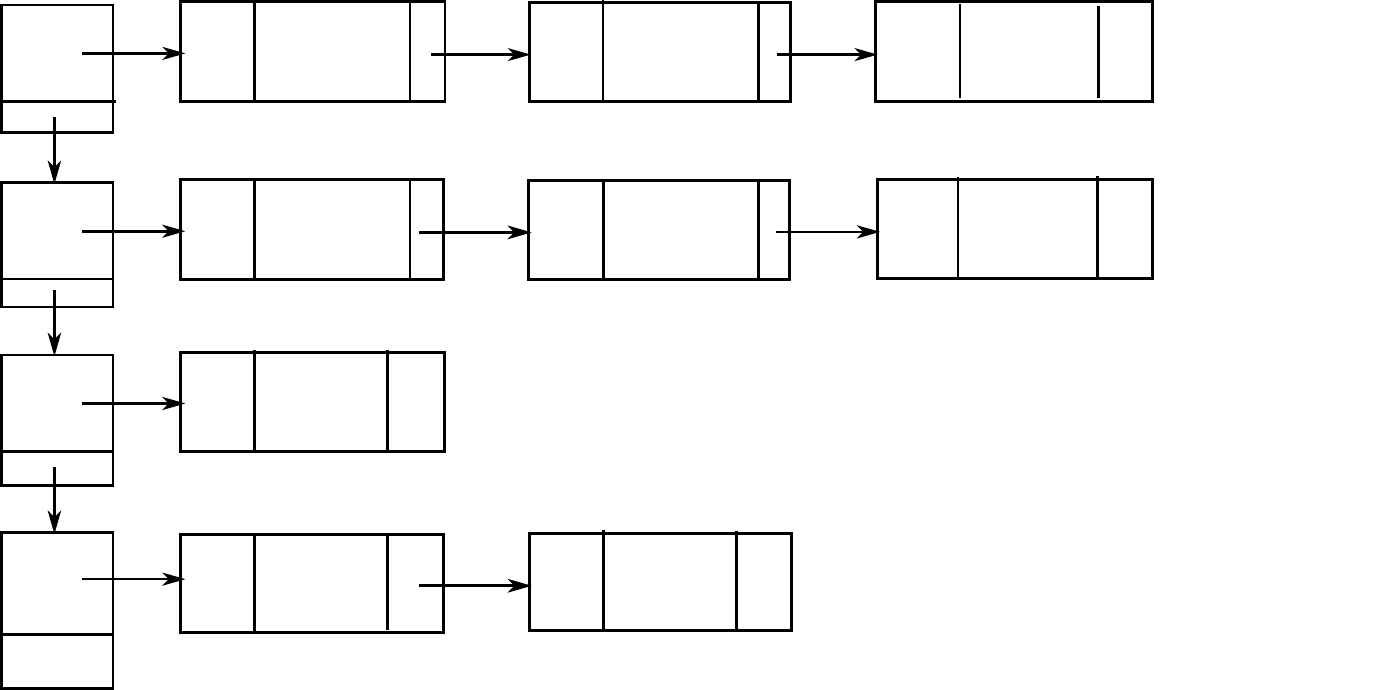
  	\caption*{(b)}
	\end{minipage}
\end{figure}
\captionof{figure}{(a)The constraint graph for Hamiltonian $H = \Pi_{12} + \Pi_{14} + \Pi_{23} + \Pi_{44}$ where $\rank(\Pi_{44}) = \rank(\Pi_{23}) = 1, \rank(\Pi_{12}) = 2$ and  $\rank(\Pi_{14}) = 3$ using its rank-1 decomposition (b) The adjacency list representation for the constraint graph $G(H)$}
\label{fig:constraint_graph}
\end{center}

We will suppose that the input to our problem is the constraint
graph $G(H)$ of the Hamiltonian, given in the standard adjacency
list representation of weighted graphs, naturally modified for
dealing with the parallel edges as shown in Figure~\ref{fig:constraint_graph}. 
In this representation there is a
linked list of size at most $n$ containing one element for each
vertex, and the element $i$ in this list is also pointing towards a
linked list containing an element for every edge $(i,j)$ or
$(i,j,b)$. For an edge $(i,j)$, this element contains $j$, the
projector $\Pi_{ij}$ and a pointer towards the next element in the
list, for an edge $(i,j,b)$ it also contains the value $b$. The
problem Q2SAT is defined formally as follows.

\vbox{\begin{quote}
Q2SAT\\
\textbf{Input:} The constraint graph $G(H)$ of a 2-local Hamiltonian
  $H$, given in the adjacency list representation.

\textbf{Output:} A solution if $H$ is frustration free, ``$H$ is
  unsatisfiable'' if it is not.
\end{quote}
}

%====================================================================
\subsection{Simple ground states}

Our algorithm is based crucially on the following \textit{product
state theorem}, which says that any frustration-free Q2SAT
Hamiltonian has a ground state which is a product state of
single qubit and two-qubit states, where the latter only appear in
the support of rank-3 projectors. A slightly weaker claim of that
form has already appeared in Theorem~2 of \Ref{CCD+11}. The
difference here is that we specifically attribute the 2-qubits
states in the product state to rank-3 projectors. Just as in
\Ref{CCD+11}, our derivation begins with Theorem~1 of \Ref{CCD+11},
which we give below. It relies on the notion of a \textit{genuinely
entangled} state in an $n$-qubit system, which is a pure state that
is not a product state with respect to any bi-partition of the
system. Then Theorem~1 in~\cite{CCD+11} states
\begin{Prop}\label{thm:genuinely-ent}
  A $2$-local frustration-free Hamiltonian on $n$ qubits which has a
  genuinely entangled ground state always has a product ground
  state, whenever $n \geq 3.$
\end{Prop}

We will also need the following {simple fact} about
2-dimensional subspaces in $\mbC^2\otimes \mbC^2$
\begin{Fact}\label{lem:rank2product}
  Any $2$-dimensional subspace $V$ of the $2$-qubit space
  $\mbC^2\otimes \mbC^2$ contains at least one product state, which
  can be found in constant time.
\end{Fact}
\begin{proof}
  Take a basis $\{\ket{\psi}, \ket{\phi}\}$ of the two-dimensional
  subspace $V^\bot$, the orthogonal complement of $V$. Our goal is
  to find a product state $\ket{\alpha}\otimes\ket{\beta}$ such that
  $\bra{\psi}(\ket{\alpha}\otimes\ket{\beta}) =
  \bra{\phi}(\ket{\alpha}\otimes\ket{\beta}) = 0$. To that aim,
  expand in the standard basis:
  $\ket{\psi}=\sum_{ij}\psi_{ij}\ket{ij}$,
  $\ket{\phi}=\sum_{ij}\phi_{ij}\ket{ij}$, and
  $\ket{\alpha}=\sum_i\alpha_i\ket{i}$,
  $\ket{\beta}=\sum_i\beta_i\ket{i}$. Then our task is to find
  coefficients $\alpha_i$ and $\beta_i$ such that $\sum_{ij}
  \phi^*_{ij}\cdot\alpha_i\beta_j =0$ and $\sum_{ij}\psi^*_{ij}
  \cdot\alpha_i\beta_j = 0$. We can pass to a matrix notation, in
  which $\psi^*_{ij}, \phi^*_{ij}$ are the entries of $2\times 2$
  matrices $\Psi, \Phi$, and $\alpha_i, \beta_i$ are the
  coordinates of the 2-vectors $\underline{\alpha},
  \underline{\beta}$.
  In that notation, we are looking for vectors $\underline{\alpha},
  \underline{\beta}$ such that
  \begin{align}
  \label{eq:2by2-cond}
    \underline{\alpha}^T\Psi \underline{\beta}
      = \underline{\alpha}^T\Phi \underline{\beta} = 0 \,.
  \end{align}
  If the matrix $\Phi$ is singular, we pick $\underline{\beta}$
  inside its the null space, and choose $\underline{\alpha}$ such
  that $\underline{\alpha}^T\Psi \underline{\beta}=0$. Otherwise,
  when $\Phi$ is non-singular, we let $\underline{\beta}$ be an
  eigenvector of the matrix $\Phi^{-1}\Psi$, i.e.,
  $\Phi^{-1}\Psi\underline{\beta}=c\underline{\beta}$, where $c$ is
  some eigenvalue. Then
  $\Psi\underline{\beta}=c\Phi\underline{\beta}$, and therefore to
  satisfy \Eq{eq:2by2-cond}, we can choose $\underline{\alpha}$ such
  that $\underline{\alpha}^T\Phi\underline{\beta}=0$.
\end{proof}

Our product state theorem is stated as follows.
\begin{Thm}\label{thm:productstate}
  Any frustration-free 	${\rm Q2SAT}$ Hamiltonian has a ground state which is
  a tensor product of one qubit and two-qubit states, where
  two-qubit states only appear in the support of rank-$3$ projectors.
\end{Thm}

%XXXXXXXXXXXXXXXXXXXXXXXXXXXXX productstate-proof2_short XXXXXXXXXXXXXXXXXXXXXXXXXX
\begin{proof}
  Consider a frustration-free Q2SAT Hamiltonian $H$ and let
  $\ket{\gs}$ be its ground state. Generally, $\ket{\gs}$ can be
  written as a product state 
  \begin{align*}
    \ket{\gs} = \ket{\alpha^{(1)}}\otimes
      \ket{\alpha^{(2)}}\otimes\cdots\otimes\ket{\alpha^{(r)}} \,,
  \end{align*}
  where each $\ket{\alpha^{(i)}}$ is a genuinely entangled state
  defined on a subset $S^{(i)}$ of qubits. Notice if
  $\Pi_{ij}=\Id-\ket{\psi}\bra{\psi}_{jk}$ is a rank-3 projector 
  then necessarily \emph{every} ground state of $H$ will contain
  $\ket{\psi}_{jk}$ at a tensor product with the rest of the system.
  Therefore, if $\ket{\psi}_{jk}$ happens to be entangled, there
  would necessarily be a subset $S^{(i)}=\{j,k\}$ in the above
  decomposition with $\ket{\alpha^{(i)}}=\ket{\psi}_{jk}$. On the
  other hand, if $\ket{\psi}_{jk}$ is a product state, there would
  be two subsets $S^{(i_1)}=\{j\}$, and $S^{(i_2)}=\{k\}$.
    
  Let $H^{(i)}$ be the Hamiltonian that is the sum of all the
  projectors whose support is in $S^{(i)}$. Clearly,
  $\ket{\alpha^{(i)}}$ is a ground state of $H^{(i)}$.  By the
  reasoning in the paragraph above, it is clear that for subsets
  $S^{(i)}$ with two or more qubits, \emph{that do not correspond to
  the support of rank-3 projectors}, the corresponding $H^{(i)}$
  consists only of rank-1 and rank-2 projectors.  It is easy to see
  that Fact~\ref{lem:rank2product} implies that for such
  Hamiltonians, which do not contain rank-3 projectors,
  Proposition~\ref{thm:genuinely-ent} also for $n=2$ case.
  Therefore, any such $H^{(i)}$ also has a ground state
  $\ket{\beta^{(i)}}$ which is a product state of one qubit states:
  \begin{align}
    \ket{\beta^{(i)}} \defeq \ket{\beta^{(i)}_1}\otimes
    \ket{\beta^{(i)}_2}\otimes \cdots
  \end{align}
  The remaining $S^{(i)}$ subsets correspond either to one qubit
  subsets, or to 2-qubits subsets of entangled rank-3 projectors. In
  all these cases, we define
  $\ket{\beta^{(i)}}\defeq\ket{\alpha^{(i)}}$.

  We now claim that the state $\ket{\beta} =
  \ket{\beta^{(1)}}\otimes\cdots\otimes \ket{\beta^{(r)}}$, which is
  a product of one-qubit and two-qubit states, is a ground state of
  $H$. To prove this we need to show that this state is in the
  ground space of every projector $\Pi_e$ in $H$. If the support of
  $\Pi_e$ is inside one of the $S_i$ subsets, then by definition
  $\Pi_e\ket{\beta^{(i)}}=0$ and therefore also
  $\Pi_e\ket{\beta}=0$. Assume then that $\Pi_e$ is supported on a
  qubit from $S^{(i)}$ and a qubit from $S^{(j)}$ with $i\neq j$. We
  now consider 3 cases:
  \begin{enumerate}
    \item If both $S^{(i)}$ and $S^{(j)}$ contain only one qubit
      then $\Pi_e\ket{\beta^{(i)}}\otimes\ket{\beta^{(j)}} = 
      \Pi_e\ket{\alpha^{(i)}}\otimes\ket{\alpha^{(j)}}=0$.
	
    \item If $S^{(i)}$ is made of one qubit but $S^{(j)}$ has two or
      more qubits, then consider the Schmidt decomposition
      $\ket{\alpha^{(j)}} = \lambda_1\ket{x_1}\otimes\ket{y_1} +
      \lambda_2\ket{x_2}\otimes\ket{y_2}$. Here,
      $\ket{x_1},\ket{x_2}$ are defined on the qubit of $S_j$ that
      is in the support of $\Pi_e$, while $\ket{y_1}, \ket{y_2}$ are
      defined on the rest of the qubits in $S_j$.
      The Schmidt coefficients $\lambda_1,\lambda_2$ are 
      by
      assumption non-zero, as $\ket{\alpha_j}$ is
      entangled. Then the condition
      $\Pi_e\ket{\alpha^{(i)}}\otimes\ket{\alpha^{(j)}}=0$ is
      equivalent to $\lambda_1\ket{y_1}\otimes
      \big(\Pi_e\ket{\alpha^{(i)}}\otimes\ket{x_1}\big) +
      \lambda_2\ket{y_2}\otimes
      \big(\Pi_e\ket{\alpha^{(i)}}\otimes\ket{x_2}\big)=0$, and by
      the linear independence of $\ket{y_1},
      \ket{y_2}$, we conclude that
      $\Pi_e\ket{\alpha^{(i)}}\otimes\ket{x_1}= 
      \Pi_e\ket{\alpha^{(i)}}\otimes\ket{x_2}=0$. Therefore, $\Pi_e$
      annihilates the subspace $\ket{\alpha^{(i)}}\otimes\mbC^2$ of
      the two qubits that it acts on, and in particular it
      annihilates $\ket{\beta^{(i)}}\otimes\ket{\beta^{(j)}}$ since
      $\ket{\beta^{(i)}}=\ket{\alpha^{(i)}}$. 
        
    \item The third case in which  both $S^{(i)}$ and $S^{(j)}$ 
      contain two or more qubits cannot happen. Indeed, in such case
      we write both $\ket{\alpha^{(i)}}, \ket{\alpha^{(j)}}$ in
      their Schmidt decomposition, and from a similar argument that
      was used above, we conclude that $\Pi_e$ must annihilate 4
      independent vectors. It therefore cannot be a rank-1 or a
      rank-2 projector.
  \end{enumerate} 
  This completes the proof of the theorem.
\end{proof}

%\label{thm:solutionform}
 
%====================================================================
\subsection{Assignments}
 
Let $H = \sum_{e \in I}$ be a 2-local Hamiltonian. By
Theorem~\ref{thm:productstate}, if $H$ is frustration free then it
has a ground state which is the tensor product of 1-qubit and
2-qubit entangled states, where the latter only appear in pairs of
qubits corresponding to rank-3 projectors. To build up a ground
state of such form, our algorithm will use partial assignments, or
shortly assignments. An \emph{assignment} $s$ is a mapping from
$[n]$. For every $i \in [n]$, the value $s(i)$ is either a 1-qubit
state $\ket{\alpha}$, or a 2-qubit entangled state
$\ket{\gamma}_{ij}$ for some $j \neq i$, or a symbol from the set
$\{\Empty, \Bad\}$. If $s(i) = \ket{\alpha}$ or $s(i) =
\ket{\gamma}_{ij}$, then this value is assigned to qubit variable
$i$, and in the latter case the entangled state is shared with
variable $j$.  The symbol $\Empty$ is used for unassigned variables,
and the symbol $\Bad$ is used when several values are assigned to
some variable.

We define the \emph{support} of $s$ by $\supp(s) = \{i \in [n] :
s(i) \neq \Empty\}$. The assignment $s$ is \emph{empty} if $\supp(s)
= \emptyset$.  
%\mnote{Why shouldn't we say, as it is in short10, that $\Empty$ also denotes the empty assignment?}
%\mnote{I propose to add: When there is no danger of confusion, we will denote the empty assignment also by $\Empty$.}
When there is no danger of confusion, we will denote the empty assignment also by $\Empty$.
We say that an assignment is \emph{coherent} if for
every $i$, we have $s(i) \neq \Bad$, and whenever $s(i) =
\ket{\gamma}_{ij}$, we also have $s(j) = \ket{\gamma}_{ji}$.
For coherent assignments $s$ and $s'$, we say that $s'$ is an
\emph{extension} of $s$, if for every $i$, such that $s(i) \neq
\Empty$, we have $s'(i) = s(i)$. A coherent assignment is
\emph{total} if $s(i) \neq \Empty$, for all $i$.
Clearly, a coherent assignments defines a product state of 1-qubit
and 2-qubits states on qubits in its support. We denote this state
by $\ket{s}$. We say that a coherent assignment $s$ {\em satisfies} a
projector $\Pi_e$, or simply that it satisfies the edge $e$, 
if for any total extension $s'$ of $s$ we have $\Pi_e\ket{s'}=0$.
%\mnote{This is inconsistent with the rest of the paper since in the algorithms we
%use the definition up to which an assignment satisfies a product rank-1 projector also when only one of the corresponding
%variables is assigned, the other not, but for every possible assignment to the other,
%the constraint is satisfied. This is detailed in short10, I think that we should put it back}

For $H = \sum_{e \in I} \Pi_e$ given in rank-1 decomposition, and a
coherent $s$, we define the \emph{reduced Hamiltonian} $H_s$ of $s$
as
\begin{align*}
  H_s = H - \sum_{s {\rm ~ satisfies ~ } e} \Pi_{e}.  
\end{align*}
We will denote the constraint graph $G(H_s)$ of the reduced
Hamiltonian $H_s$ by $G_s = (V_s, E_s)$.  We call a coherent
assignment $s$
a \emph{pre-solution} if it has a total extension $s'$ 
satisfying every constraint in $H$,  and we call $s$ a \emph{solution} if
$s$ itself satisfies every constraint in $H$.
%\inote{Why???} 
%\mnote{It is a definition, it seems to be  the right one.
%Again, we are now in the middle of changing many things, and we have to remain consistent.
%I can imagine that we don't define {\em satisfaction} for a state, only for an assignment.
%And then we should use exactly this definition, except replacing $\ket{s}$ by $s$}
Obviously, an
assignment is a solution if and only if $G_s$ is the empty graph.  
%\inote{All the above paragraph is really
%tiresome and difficult to follow. Do we really need it???}
%\mnote{This is of course related to the problem raised above, I wanted to output at the end a state which is defined for every qubit,
%even when a satisfying assignment was found for less than $n$ qubits}
A coherent assignment $s$ is \emph{closed} if $\supp(s) \cap V_s = \emptyset$.

%XXXXXXXXXXXXXXXXXXXXXXXXXXXXX propagation_short XXXXXXXXXXXXXXXXXXXXXXXXXX

%%%%%%%%%%%%%%%%%%%%%%%%%%%%%%%%%%%%%%%%%%%%%%%%%%%%%%%%%%%%%%%%%
\section{Propagation}

The crucial building block of our algorithm is the propagation of
values by rank-1 projectors.  This is the quantum analog of the
classical propagation process when for example the clause $x_i \vee
x_j$ propagates the value $x_i= 0$ to the value $x_j = 1$ in the
sense that given $x_i = 0$, the choice $x_j = 1$ is the only
possibility to make the clause true. In the quantum case this notion
has already appeared in \Ref{ref:Laumann09}, and can in fact be
traced back also to Bravyi's original work. Here, we shall adopt the
following definition
\begin{Def}[Propagation]
  Let $\Pi_e = \ket{\psi}\bra{\psi}$ be a rank-$1$ projector acting
  on variables $i,j$, and let $\ket{\alpha}$ be either a $1$-qubit
  state assigned to variable $i$, or a $2$-qubit entangled state
  assigned to variables $k,i$ for some $k\neq j$. We say that
  $\Pi_e$ \emph{propagates} $\ket{\alpha}$ if, up to a phase, there
  exists a unique $1$-qubit state $\ket{\beta}$ such that $\Pi_e
  \ket{\alpha}\otimes\ket{\beta}_j = 0.$ In such case we say that
  $\ket{\alpha}$ is propagated to $\ket{\beta}$ along $\Pi_e$, or
  that $\Pi_e$ propagated $\ket{\alpha}$ to $\ket{\beta}$.
\end{Def}
The following lemma shows how the propagation properties of 
$\Pi_e=\ket{\psi}\bra{\psi}$ are determined by entanglement in
$\ket{\psi}$.
\begin{Lem}
\label{lem:propagate0}
  Consider the rank-$1$ projector $\Pi_e = \ket{\psi}\bra{\psi}$,
  defined on qubits $i,j$.  If $\ket{\psi}$ is entangled, it
  propagates \emph{every} $1$-qubit state $\ket{\alpha}_i$ to a
  state $\ket{\beta(\alpha)}_j$ such that if
  $\ket{\alpha}_i\neq\ket{\alpha'}_i$ then $\ket{\beta(\alpha)}_j
  \neq \ket{\beta(\alpha')}_j$. This propagation can be calculated
  in constant time. When $\ket{\psi}$ is a product state
  $\ket{\psi}=\ket{x}_i\otimes\ket{y}_j$, the projector $\Pi_e$ does
  not propagate states that are proportional to $\ket{x^\perp}_i$,
  while all other states are propagated to $\ket{y^\perp}_j$.
\end{Lem}

\begin{proof}
  Assume that $\ket{\psi}$ is entangled and consider the state
  $\ket{\alpha}$. Our task is to show that there always exists a
  unique $\ket{\beta}$ (up to an overall constant) such that 
  $\Pi(\ket{\alpha}\otimes \ket{\beta})=0$, and that different
  $\ket{\alpha}$'s yield different $\ket{\beta}$'s. 
  
  Expanding $\ket{\psi}$, $\ket{\alpha}$, and $\ket{\beta}$ in the
  standard basis $\ket{\psi}=\sum_{i,j}\psi_{ij}
  \ket{i}\otimes\ket{j}$; $\ket{\alpha}=\sum_i\alpha_i\ket{i}$;
  $\ket{\beta}=\sum_j\beta_j\ket{j}$, the condition 
  $\Pi_e(\ket{\alpha}\otimes\ket{\beta}) =0$ translates to 
  $\sum_{i,j}\psi^*_{ij}\alpha_i\beta_j=0$.  Assuming that
  $\ket{\psi}$ is entangled, one can easily verify that the $2\times
  2$ matrix $(\psi^*_{ij})$ is non-singular. Then using the simple
  fact that in a two-dimensional space every non-zero vector has
  exactly one non-zero vector (up to an overall scaling) which it is
  orthogonal to, it is straightforward to deduce that for every
  non-zero vector $(\alpha_0, \alpha_1)$ there is a unique (up to
  scaling) non-zero vector $(\beta_0, \beta_1)$ such that
  $\sum_{i,j}\psi^*_{ij}\alpha_\mu\beta_\nu=0$. Moreover, 
  $(\beta_0,\beta_1)$ can be calculated in constant time, and that
  different $(\alpha_0, \alpha_1)$ necessarily yield different
  $(\beta_0, \beta_1)$.
  
  The case when $\ket{\psi}$ is a product state is straightforward.
\end{proof}

We now present two lemmas that describe the structure of the \emph{global}
ground state of the system, if we know that part of it is in a
tensor product of 1-qubit or 2-qubits states, which are then
propagated by some $\Pi_e$.
\begin{Lem}[Single qubit propagation]
\label{lem:propagate1}

  Consider a frustration-free ${\rm Q2SAT}$ system $H=\sum_{e \in I}
  \Pi_e$ with a rank-$1$ projector $\Pi_e = \ket{\psi}\bra{\psi}$
  between qubits $i,j$, and assume that $H$ has a ground state of
  the form $\ket{\gs}=\ket{\alpha}_i\otimes\ket{rest}$. Then:
  \begin{enumerate}
    \item If $\Pi_e$ propagates $\ket{\alpha}_i$ to $\ket{\beta}_j$
      then necessarily $\ket{rest}=\ket{\beta}_j\otimes\ket{rest'}$.
    \item $\ket{\gs}$ is also a ground state of the 
      ${\rm Q2SAT}$ Hamiltonian $H-\Pi_e$.
  \end{enumerate}
\end{Lem}

\begin{proof}
  For the first claim assume that $\Pi_e$ propagates
  $\ket{\alpha}_i$ to $\ket{\beta}_j$. Without loss of generality,
  we may expand
  \begin{align*}
    \ket{rest} =\ket{\beta}_j\otimes\ket{rest_1} +
      \ket{\beta^\perp}_j\otimes\ket{rest_2} \ ,
  \end{align*}
  where the states $\ket{rest_1}, \ket{rest_2}$ are defined on
  all the qubits of the system except for $(i,j)$, and are not
  necessarily normalized. Plugging this expansion into the condition
  $\Pi_e\ket{\gs}=0$, we obtain the equation
  \begin{align*}
    (\Pi_e\ket{\alpha}_i\ket{\beta}_j)\otimes \ket{rest_1}
     + (\Pi_e\ket{\alpha}_i\ket{\beta^\perp}_j)\otimes \ket{rest_2}
     =0 \,.
  \end{align*}
  Since $\Pi_e$ propagates $\ket{\alpha}_i$ to $\ket{\beta}_j$, we
  have $\Pi_e \ket{\alpha}_i\ket{\beta}_j=0$ and $\Pi_e
  \ket{\alpha}_i\ket{\beta^\perp}_j\neq 0$. Therefore, the
  above equation implies that $\ket{rest_2}=0$, and we may set
  $\ket{rest'}=\ket{rest_1}$.
  
  The second claim follows trivially from the frustration-freeness
  of the system.
\end{proof}

\begin{Lem}[Entangled 2-qubits propagation]
\label{lem:propagate2}
  Consider a frustration-free ${\rm Q2SAT}$ system $H$ with a rank-$1$
  projector $\Pi_e = \ket{\psi}\bra{\psi}$ between qubits $i,j$.
  Assume that $H$ has a ground state of the form
  $\ket{\gs}=\ket{\phi}_{ik}\otimes\ket{rest}$, where $\ket{\phi}$
  is an entangled state on qubits $i,k$ with $k\neq j$.  Then:
  \begin{enumerate}
    \item $\ket{\psi}$ is % necessarily 
    a product state $\ket{\psi} 
      =\ket{x}\ket{y}$.
      
    \item $\Pi_e$ propagates $\ket{\phi}$ to $\ket{y^\perp}$ and 
     % up to a phase, 
      necessarily $\ket{rest} =
      \ket{y^\perp}_j\otimes\ket{rest'}$.
    
     \item $\ket{\gs}$ is also a ground state of the ${\rm Q2SAT}$ 
       Hamiltonian $H-\Pi_e$.
  \end{enumerate}
\end{Lem}

\begin{proof}
  Write $\ket{\phi}$ in its Schmidt decomposition $\ket{\phi} =
  \lambda_1\ket{\alpha}\otimes\ket{\beta} +
  \lambda_2\ket{\alpha^\perp}\otimes\ket{\beta^\perp}$, and note
  that both $\lambda_1,\lambda_2\neq 0$, since $\ket{\phi}$ is
  entangled. Plugging this into the condition $\Pi_e\ket{\gs}=0$, we
  get
  \begin{align*}
     \Pi_e\ket{\gs}=\lambda_1\ket{\beta}_k\otimes \Pi_e
       \big(\ket{\alpha}_i\otimes \ket{rest}\big)
     +\lambda_2\ket{\beta^\perp}_k\otimes \Pi_e
       \big(\ket{\alpha^\perp}_i\otimes \ket{rest}\big) =0 \,.
  \end{align*}
  Since $\ket{\beta}$ is is linearly independent of
  $\ket{\beta^\perp}$, we conclude that
  $\Pi_e\big(\ket{\alpha}_i\otimes \ket{rest}\big) =
  \Pi_e\big(\ket{\alpha^\perp}_i\otimes \ket{rest}\big) =0$.
  
  To prove the first claim, assume by contradiction,
  that $\ket{\psi}$ is entangled. Then by
  Lemma~\ref{lem:propagate0}, $\Pi_e$ propagates $\ket{\alpha}$ and
  $\ket{\alpha^\perp}$ to two \emph{different} states, say,
  $\ket{\gamma_1}\neq \ket{\gamma_2}$. But then by
  Lemma~\ref{lem:propagate1}, it follows that $\ket{rest}$ must be
  both in the form $\ket{\gamma_1}_j\otimes\ket{rest'}$ and
  $\ket{\gamma_2}_j\otimes\ket{rest'}$ -- which is a contradiction!
  
  For the second claim, assume that
  $\ket{\psi}=\ket{x}\otimes\ket{y}$ is a product state. Since
  $\Pi_e\big(\ket{\alpha}_i\otimes \ket{rest}\big) =
  \Pi_e\big(\ket{\alpha^\perp}_i\otimes \ket{rest}\big) =0$, both
  states $\ket{\alpha}_i\otimes \ket{rest},
  \ket{\alpha^\perp}_i\otimes \ket{rest}$ are ground states of the
  single projector Hamiltonian $\tilde{H}=\Pi_e$. Using
  Lemma~\ref{lem:propagate0} and Lemma~\ref{lem:propagate1},
  together with the fact that 
  and at
  least one of the states $\ket{\alpha}, \ket{\alpha^\perp}$ is
  different from $\ket{x^\perp}$, we conclude that 
  $\ket{rest}=\ket{y^\perp}_j\otimes\ket{rest'}$. 
    
  The third claim, as before, follows simply from the
  frustration-freeness of the system.
\end{proof}

Let $H$ be a Q2SAT Hamiltonian in rank-1 decomposition, let
$s$ be a coherent assignment, % such that $s(i) \neq \Empty$, 
and let $G_s = (V_s, E_s)$ be the constraint graph of the reduced
Hamiltonian $H_s$. We would like to describe in $G_s$ the result of
the \emph{iterated propagation} process when a value given to variable $i$ is % in $s(i)$ is
propagated along all possible projectors, then the propagated values
are propagated on their turn, and so on until no more value assigned
during this process can be further propagated. 
The propagation can get started when the initial value is already assigned by $s$, that is  when 
$s(i) = \ket{\delta}$  for $\ket{\delta} \in \{ \ket{\alpha}, \ket{\gamma}_{ij}\}$,
where $\ket{\alpha}$ is some 1-qubit state and $\ket{\gamma}_{ij}$ some a 2-qubit state,
or it can get started when $s(i) = \Empty$, in which case we shall explicitly
choose a 1-qubit state $\ket{\alpha}$ 
and assign it to $i$.
%We will deal simultaneously with the situations when in $V_s \cap \supp(s) = \{i\}$ or $|V_s \cap \supp(s)| > 1$.

Let now $s, i$ and $\ket{\delta}$ be such that $s(i) \in \{\Empty , \ket{\delta} \}$.
We say that in the constraint graph $G_s$ an edge $e \in E_s$ from
$i$ to $j$ \emph{propagates} $\ket{\delta}$ if $\Pi_e$ propagates it, and
we denote 
by $\pro(s, e, \ket{\delta})$ the state $\ket{\delta}$ is propagated to.
We generalize the notion of propagation in $G_s$ from edges to
paths. Let $i = i_0, i_1, \ldots i_k$ be vertices in $V_s$, and let
$e_j$ be an edge from $i_j$ to $i_{j+1}$, for $j = 0, \ldots, k-1$.
Let $s(i) \in \{\Empty , \ket{\delta} \}$, and set $\ket{\alpha_0} = \ket{\delta}$. 
Let $\ket{\alpha_1}, \ldots, \ket{\alpha_k}$
be states such that the propagation of $\ket{\alpha_j }$ along
$\Pi_{e_j}$ is $ \ket{\alpha_{j+1}}$, for $j = 0, \ldots, k-1$. Then
we say that the path $p = (e_0, \ldots, e_{k-1})$ from $i_0$ to
$i_k$  \emph{propagates} $\ket{\delta}$, and we set $\pro(s,p, \ket{\delta}) =
\ket{\alpha_k}$. We say that a vertex $j \in V_s$ is
\emph{accessible} by propagating $\ket{\delta}$ from $i$ if either $j=i$ or there
is a path from $i$ to $j$ that propagates $\ket{\delta}$. We denote by $V^{\pro}_s (i, \ket{\delta})$
the set of such vertices, and by $\extt^{\pro}_s (i, \ket{\delta})$ the extension of $s$ by
the values given to the vertices in $V^{\pro}_s (i, \ket{\delta})$ by iterated propagation.
%$\ket{\delta}$ at $i$, if it is applicable, and all the values assigned by the propagating paths.

Let us suppose that $s' = \extt^{\pro}_s (i, \ket{\delta})$ is also coherent. The set
$V^{\pro}_s (i, \ket{\delta})$ divides the edges $E_s$ into three disjoint
subsets: the edges $E_1$ of the induced subgraph
$G(V^{\pro}_s (i, \ket{\delta}))$, the edges $E_2$ between the induced subgraphs
$G(V^{\pro}_s (i, \ket{\delta}))$ and $G(V_s \setminus V^{\pro}_s (i, \ket{\delta}))$, and the
edges $E_3$ of the induced subgraph $G(V_s \setminus
V^{\pro}_s (i, \ket{\delta}))$. While the edges in $E_1 \cup E_2$ are satisfied by
$s'$, none of the edges in $E_3$ is satisfied. Therefore $G_{s'}$ is
nothing but $G(V_s \setminus V^{\pro}_s (i, \ket{\delta}))$ without the
isolated vertices, and it can be constructed by the following
process. Given $s$ and $i$, the edges in $E_1 \cup E_2$ can be
traversed via a breadth first search rooted at $i$. The levels of
the tree are decided dynamically: at any level the next level is
composed of those vertices whose value is propagated from the
current level. The leaves of the tree are vertices in $V_s \setminus
V^{\pro}_s (i, \ket{\delta})$. The algorithm {\sf Propagation} uses a temporary
queue $Q$ to
implement this
process.

\begin{algorithm}
	\caption{{\sf Propagation}$(s, G_s, i, \ket{\delta})$}
	\label{alg:prop0}
	\begin{algorithmic}[H]
	\State
	$s(i) := \ket{\delta}$
	%\State
	%{\bf remove} from $E_s$ every edge 
	\State
	{\bf create} a list $L$ and a queue $Q$, and 
	{\bf put} $i$ into $Q$
	%$s(i) := \{\ket{\alpha}\}$
	\While {$s$ is coherent  and $Q$ is not empty} %some $j \in V^{\pro}_{s,i}$ is not visited, starting with $i$, in level-order}
		%\LState
		%\For {all vertex $j \in V^{\pro}_{s,i}$ starting with $i$, in level-order } 
		\State {\bf remove} the head $j$ of $Q$ % and {\bf mark} it {\em visited}
		%and {\bf traverse} the adjacency list of $j$
		\For {all edge $e$ from $j$ to $k$ }
		%{\bf if} $e$ is unpropagated $s(k) := s(k) \cup \{\pro(s(j),e)\}$
		\State
		{\bf remove} $e$ from $E_s$

		\If  {$e$ propagates $s(j)$}
		\State
		 %$s(k) := s(k) \cup \{\pro(s(j),e)\}$ 
		 $s(k) :=
\begin{cases}
\pro(s, e, s(j))  & \text{if ~$ s(k) = \Empty$}\\
\Bad & \text{if ~$ s(k) \not\in \{\Empty, \pro(s, e, s(j)) \}$ }
\end{cases}
$

		 \State
		 {\bf enqueue} $k$
		 \EndIf
		\If {$e$ is not propagating and $k$ is not in $L$} %{\bf and} $k$ is isolated}
		{\bf put} $k$ into $L$ %from $V_s$
		\EndIf
		\EndFor
		\State
		{\bf remove} $j$ from $V_s$
		%\State
		%{\bf remove} all isolated vertices from $V_s$
		%\EndFor
		\EndWhile
		\State
		{\bf if} {$s$ is not coherent }
		\Return ``unsuccessful"
		\For {all $k$ in $L$}
		\For {all edges $e$ from $k$ to $\ell$}
		\If {$\ell$ was removed from $V_s$} {\bf remove} $e$
		\EndIf
		\EndFor
		\If {all edges outgoing from $k$ were removed}
		{\bf remove} $k$ from $V_s$
		\EndIf
		\EndFor
		
		%\Else 
		%~ {\Return $(s, G(H_s))$}
		%\EndIf

	%	\LState \Return 0

	\end{algorithmic}
\end{algorithm}

\begin{Lem}$({\bf Propagation ~  Lemma})$
\label{lem:propagate}

  Let {\sf Propagation}$(s, G_s, i, \ket{\delta})$ be called when $H_s$ doesn't
  have rank-$3$ constraints, and  $s(i) \in \{\Empty , \ket{\delta} \}$.
  Let $s'$
  and $G' = (V',E')$ be the outcome of the procedure. Then:
  \begin{enumerate}
    \item If {\sf Propagation}$(s, G_s, i, \ket{\delta})$ doesn't 
      return ``unsuccessful" then $s' = \extt^{\pro}_s (i, \ket{\delta})$ and $G' = G_{s'}$.
      Moreover, if $s$ is
      a pre-solution then $s'$ is a pre-solution, and if $s$ is closed then $s'$ is also closed.
      
    \item If {\sf Propagation}$(s, G_s, i)$ returns ``unsuccessful" 
      then there is no solution $z$ of which is an extension of $s$ and for which $z(i) = \ket{\delta}$.
      
    \item The complexity of the procedure is $O( |E_s| - | E_{s'}|$).
  \end{enumerate}
\end{Lem}

\begin{proof}
The assignments made during the breadth first search correspond exactly to the
the paths propagating $\ket{\delta}$ from $i$, therefore the extension of $s$ created by the process is
indeed $s' = \extt^{\pro}_s (i, \ket{\delta})$. %To prove that $G' = G_{s'}$, l
%Let us denote by $F$ the set of edges in $E_s$ which are satisfied by $s'$. 
  The while loop removes the edges 
  between vertices in $V^{\pro}_s (i, \ket{\delta})$ and the edges which go from $V^{\pro}_s (i, \ket{\delta})$ to
  $V_s \setminus V^{\pro}_s (i, \ket{\delta})$, as well as the vertices  in $V^{\pro}_s (i, \ket{\delta})$.
 % of $F$ from $E_s$ and the vertices in $V^{\pro}_s (i, \ket{\delta})$. 
 Then the edges from $V_s \setminus V^{\pro}_s (i, \ket{\delta})$ to $V^{\pro}_s (i, \ket{\delta})$
 are removed, as well as the remaining vertices without outgoing (and incoming) edges.
 %
%  $V_s \setminus V^{\pro}_s (i, \ket{\delta})$
%which become isolated  isolated by the removal of the edges in $F$, 
Therefore we have $G' = G_{s'}$. 
  
%  which are adjacent to some vertex in $V^{\pro}_s (i, \ket{\delta})$. The
%  procedure visits and removes the vertices in $V^{\pro}_s (i, \ket{\delta})$ and
%  traverses and removes the edges in $F$ since they are satisfied.
%  This implies that $E' = E_s \setminus F$. When an edge $e$ from
%  $i$ to $j$ is traversed, if its projector $\Pi_e$ is still
%  unsatisfied, it becomes satisfied because $s(i)$ is propagated
%  along $e$ in order to exactly do that. Thus the extension $s'$
%  satisfies every $\Pi_e$ for $e \in F$, in addition to the
%  projectors already satisfied by $s$, and therefore $E_{s'} = E_s
%  \setminus F = E'$. Since the procedure removes every vertex which
%  becomes isolated in $V_s$, we also have $V' = V_{s'}$, and
% therefore $G' = G_{s'}$. 

Let us suppose that $s$ is a pre-solution, and let $z$ be an extension of $s$
which is a solution and which is a product state on the vertices in $V_s$. By Theorem~\ref{thm:productstate}
there exists such a solution since $H_s$ doesn't have rank-3 constraints. We define the assignment $z'$ by
$$z'(j) =
\begin{cases}
s'(j)  & \text{if ~$ j \in \supp(s')$}\\
z(j) & \text{otherwise.}
\end{cases}
$$
Then $z'$ is a solution which is an extension of $s'$, and therefore $s'$ is a pre-solution.
 If $s$ is closed then
so is  $s'$ since only the vertices in $V^{\pro}_s (i, \ket{\delta})$ 
get assigned during the process, and they are not included into $V_{s'}$.
% is closed since every vertex in $V^{\pro}_s (i, \ket{\delta})$ gets assigned. 

%  If $s$ is a pre-solution then we can show that at every extension step,
%  the extended assignment is also a pre-solution. Indeed, let us suppose that
%  at some point in the procedure, a pre-solution $s_1$ is 
% extended by propagation into a coherent assignment $s_2$, where the newly assigned
% variable is $k$, and $s_2(k)= \ket{\beta}$.
%  Then either a single qubit, or an entangled 2-qubit state was propagated 
%  from some variable to $k$.
%  Thus Lemma~\ref{lem:propagate1}
%  or Lemma~\ref{lem:propagate2} respectively imply that any solution $s^*$ which is a total extension of 
%  $s_1$ satisfies $s^*(k) = \ket{\beta}$, and therefore $s^*$ is also a total extension of $s_2$.
%  This argument can be repeated at every propagation step, and thus $s'$ is also a pre-solution.
  %  If $s$ is a pre-solution, it follows by repeated applications of Lemmas~\ref{lem:propagate1} 
  %and~\ref{lem:propagate2} that $s'$ is also a pre-solution.
 
  Let us now suppose that the procedure returns ``unsuccessful".
  Then there is a vertex  $k \in V^{\pro}_s (i, \ket{\delta})$, and two paths
  $p$ and $p'$ in $G_s$ from $i$ to $k$ such that 
  $\pro(s,p, \ket{\delta}) = \ket{\beta}$, $\pro(s,p', \ket{\delta}) = \ket{\beta'}$
  and $\ket{\beta} \neq \ket{\beta'}$. 
  Let us also suppose that there exists a solution $z$ which is
  an extension of $s$ and for which $z(i) = \ket{\delta}$.
  Then by the repeated use of Lemma~\ref{lem:propagate1},
  and also by using once Lemma~\ref{lem:propagate2} when $\ket{\delta}$ is a 2-qubit entangled state,
  we conclude that $z(k)$ is simultaneously equal to $\ket{\beta}$ and to $\ket{\beta'}$, which is a contradiction.

%   to which the procedure
%  assigned two different values, respectively propagated along some
%  paths $p$ and $p'$ in $G_s$. 
%  
%  Then, by an argument similar to the one above,
%  we know by Lemmas~\ref{lem:propagate1} and~\ref{lem:propagate2}
%  that no product state
%  extension of $s$ can satisfy simultaneously every constraints on these paths.
%  But since $H_s$ doesn't have rank-3 constraints, by
%  Theorem~\ref{thm:productstate} if it is frustration free, it has
%  a product ground state. Therefore $H_s$ is not frustration free. 
  
  Finally Statement 3 follows since
  every step of the procedure can be naturally charged to an edge in $E_s \setminus E_{s'}$,
  and every edge is charge only a constant times.

\end{proof}

%XXXXXXXXXXXXXXXXXXXXXXXXXXXXX full-algo-short XXXXXXXXXXXXXXXXXXXXXXXXXX

\section{The main algorithm}

% token holder: Shengyu
\subsection{Description of the algorithm} 

We now give in broad lines a description of our algorithm we call
{\sf Q2SATSolver}. It takes as input the the adjacency list
representation of the constraint graph $G(H)$ of a 2-local
Hamiltonian $H$ in rank-1 decomposition. The algorithm uses four
global variables: assignments $s_0$ and $s_1$ initialized to
$\Empty$, and graphs $G_0$ and $G_1$ in the adjacency list
representation, initialized to $G(H)$. The algorithm consists of
four phases, and except the first one, each phase consists of
several stages, where essentially one stage corresponds to one {\sf
Propagation} process. In the case of an unsatisfiable $Hamiltonian$ the
algorithm at some point outputs ``$H$ is unsatisfiable" and stops.
This happens when either the maximal rank constraints are already
unsatisfiable, or at some later point several values are assigned to the same variable 
during a necessary propagation process. 

In the case of a frustration-free Hamiltonian,
at the beginning and end of each stage, we will
have $s_0=s_1$, f%or consistency, 
and $G_0 =G_1 =G_{s_0}$. In the first two
phases only $(s_0, G_0)$ develops, and is copied to $(s_1, G_1)$ at the 
end of the phase. In the last two phases, $(s_0, G_0)$ and $(s_1, G_1)$ develop
independently, but only the result of one of the two processes is retained and is 
copied into the other variable at the end of the phase. This parallel development of the two processes
is necessary for complexity considerations, it
ensures that the useless work done is proportional to the useful work.

In the first phase the procedure {\sf MaxRankRemoval} satisfies, if this is 
possible, all constraints of maximal rank. In the second phase all these assignments
are propagated, which, if successful, results in a closed
assignment $s$ such that $H_s$ has only rank-1 constraints. In the
third phase the procedure {\sf ParallelPropagation} satisfies the product 
constraints one by one and propagates the assigned values. To
satisfy a product constraint, the only two possible choices are
tried and propagated in parallel. In the fourth phase the remaining
entangled constraints are taken care of, again, one by one. To
satisfy a constraint, an arbitrary value is tried and propagated. In
case of an unsuccessful propagation we are able to efficiently find
a product constraint implied by the entangled constraints
considered during the propagation, and therefore it becomes possible to
proceed as in phase three. In case of success we are left with a
satisfying assignment and the empty constraint graph.
Theorem~\ref{thm:intro} is an immediate consequence of the following result.
%\mnote{Something went wrong with the macros. The sign for an orthogonal state in the call of
%ParallelPropagation looks like the
%sign for the empty assignment in line.}
%\mnote{I have corrected it}

\begin{algorithm}
\caption{{\sf Q2SATSolver}$(G(H))$}
\label{alg:Q2SATfull}
%\textbf{input}: A 2-local Hamiltonian $H = \sum_{(i,j)\in E} \Pi_{ij}$.

%\textbf{output}: A state $\ket{\psi}$ %$\otimes_{e:rank(\Pi_e)=3} \ket{\psi_e} \otimes_{i:rest}\ket{s_0(i)}$ 
%satisfying $H$ if $H$ is satisfiable, or ``unsatisfiable'' if $H$ is unsatisfiable.

%\textbf{global variables}: (1) $s_0$ and $s_1$, two $n$-element arrays. (2) $n$ linked lists containing the graph and constraint information. 
%\\
\begin{algorithmic}[H]
	%\LState Run {\sf DSPrep}. \Comment{Construct the graph adjacency lists}. 
%	\LState \For {$i \in [n]$} 
%	\LState $s_0(i) := \emptyset, ~ s_1(i) := \emptyset$
%	\EndFor
	
	\State
	%{\bf for} ~ $i \in [n]$ ~ {\bf do} ~ $s_0(i) = s_1(i) := \emptyset$ \Comment{Initialize global variables}
	%\LState $G_0 =G_1 := G(H_{s_0})$
	\State $s_0 = s_1 := \Empty, ~~ G_0 = G_1 := G(H)$  \Comment{Initialize global variables}
	
	\vspace{.5em} 
	\State {\sf MaxRankRemoval()} \Comment{Remove maximal rank constraints} 
	
%	\vspace{.5em} 
%	\For {each rank-2 projector $\Pi_{ij}$} \Comment{Remove  degenerate rank-2 projectors.}
%		\If {$\Pi_{ij}$ projects on $\mbC^2_i\otimes \ket{\beta}_j$}
%			\LState $s_0(j) = \ket{\beta^\Empty}$; remove edge $(i,j)$.
%			\LState {\sf Propagation}$(j,\ket{\beta^\Empty})$;
%		\ElsIf {$\Pi_{ij}$ projects on $\ket{\alpha}_i \otimes \mbC^2_i$}
%			\LState $s_0(i) = \ket{\alpha^\Empty}$; remove edge $(i,j)$.
%			\LState {\sf Propagation}$(i,\ket{\alpha^\Empty})$;
%		\EndIf
%	\EndFor 
	\vspace{.5em} 
%\While {$s_0$ is coherent and not closed}
\While  {there exist $i \in V_0$ such that $s(i) \neq \Empty$} \Comment{Propagate all assigned values} 
\State
%~~~~{\sc Propagate}$(s,i, s(i))$ for some $i \in V_s$
{\sc Propagate}$(s_0, G_0, i, s_0(i))$ for some vertex $i$ in $G_0$ such that $s_0(i) \neq \Empty$
\State
{\bf if} the propagation returns ``unsuccessful"  {\bf output} ``$H$ is unsatisfiable"
\State
%{\sc Propagate}$(s_1, G_1, i)$ for the same vertex $i$ 
$s_1 := s_0, G_1 := G_0$
\EndWhile
	\vspace{.5em}	
	%\While {in $G_0$ there exists a product edge from $i_0$ to $i_1$ with constraint
	%$\ket{\alpha_0}_{}\otimes \ket{\alpha_1}_{}$}
	\While {there exists in $G_0$ a product edge %from $i_0$ to $i_1$ 
	with constraint
	$\ket{\alpha_0^\bot}_{i_0} \otimes \ket{\alpha_1^\bot}_{i_1} \bra{\alpha_0^\bot}_{i_0}  \otimes \bra{\alpha_1^\bot}_{i_1} $}
	\State 
	{\sf ParallelPropagation}$(i_0,\ket{\alpha_0}, i_1,\ket{\alpha_1})$ 
	\Comment{Remove product constraints}
	\EndWhile

	%\For {each product constraint $\ket{\alpha}_{u}\otimes \ket{\beta}_{v}$} \Comment{Remove product constraints.}
	%	\LState Remove this product constraint.
	%	\LState {\sf ParallelPropagation}$(u,\ket{\alpha^\Empty}, v,\ket{\beta^\Empty})$.
	%\EndFor \Comment{Only rank-1 entangled constraints remain.}

	\vspace{.5em} 
	\While {$G_0$ is not empty }  \Comment{Remove entangled constraints}
	\State 
	{\sf ProbePropagation}$(i)$ for some vertex $i$
	%\Comment{Remove product constraints.}
	\EndWhile
	\State {\bf output} $\ket{s}$ for any total extension $s$ of $s_0$.
	
	%\For {each remaining vertex $i$}
		%\LState Take an arbitrary vertex $i$ in $C$.
%		\LState %(flag, $j$, $\ket{\psi_0}$, $\ket{\psi_1}$) = 
%		{\sf ProbePropagation}$(i,\ket{0})$.
%	\If {flag = 1}
%		\LState {\sf ParallelPropagation}$(j,\ket{\psi_0}, j,\ket{\psi_1})$.
%	\EndIf
%	\EndFor

%	\vspace{.5em}
%	\If {$\otimes_{i=1}^n {s_0(i)}$ satisfies the original input Hamiltonian $H$} \Comment{Final verification}
%		\LState Output $\ket{\psi} = \Big(\otimes_{(i,j)\in E_{ent}, i<j}{s_0(i)}\Big)\otimes \Big(\otimes_{i\in V-V_{ent}} {s_0(i)}\Big)$ as a ground state of $H$.
%	\Else
%		\LState Output ``Unsatisfiable''.
%	\EndIf
\end{algorithmic}
\end{algorithm}

\begin{Thm}\label{thm:main}
  Let $G(H) = (V,E)$ be the constraint graph of a $2$-local
  Hamiltonian. Then:
  \begin{enumerate}
    \item If $H$ is frustration-free, the algorithm 
      {\sf Q2SATSolver}$(G(H))$ outputs a ground state $\ket{s}$.
      
    \item If $H$ is not frustration-free, the algorithm 
      {\sf Q2SATSolver}$(G(H))$ outputs ``H is unsatisfiable''.
      
    \item The running time of the algorithm is $O(|V| + |E|)$.
  \end{enumerate}  	
\end{Thm}
Theorem~\ref{thm:main} will be proven in
Section~\ref{subsec:analysis}.

%XXXXXXXXXXXXXXXXXXXXXXXXXXXXX rank3_short XXXXXXXXXXXXXXXXXXXXXXXXXX

%====================================================================
\subsection{Max rank removal}
The {\sf MaxRankRemoval} procedure is conceptually very simple.
Since every maximal rank constraint has a unique solution (up to a
global phase), it makes this assignment for each constraint, and
then checks if this is globally consistent.

\begin{algorithm}
\caption{{\sf MaxRankRemoval}$()$}
\label{alg:R3R}

\begin{algorithmic}[H]
\State
{\bf for} {all $i \in V_0$ such that $\rank(\Pi_{ii}) = 1$ and $\ket{\phi}$ is the unique state satisfying $\Pi_{ii}$}
\State
~~~~$s_0(i) := \ket{\phi}$ 
%~~~~Remove $(i,i)$ from $G(H)$
%\LState
%\vspace{.5em}
\For {all $i \in V_0$, for all edge $e \in E_0$ from $i$ to $j$ such that $\rank(\Pi_{e}) = 3 $ } 
\If { \text{$\ket{\alpha}_i \ket{\beta}_j $ is the unique product 2-qubit state satisfying $\Pi_{e}$}}
\State 
~~~~$s_0(i) :=
\begin{cases}
\ket{\alpha} & {\rm if} ~ s_0(i) = \Empty\\
\Bad & {\rm if}~ s_0(i) \notin \{ \Empty , \ket{\alpha}\}
%\text{if $\ket{\phi}$ is the unique entangled 2-qubit state satisfying $\Pi_{e}$}
\end{cases}
$
\EndIf
\If {\text{$\ket{\gamma}$ is the unique entangled 2-qubit state satisfying $\Pi_{e}$}}
\State 
~~~~$s_0(i) :=
\begin{cases}
\ket{\gamma}_{ij} & {\rm if} ~ s_0(i) = \Empty\\
\Bad & {\rm if}~ s_0(i) \notin \{ \Empty , \ket{\gamma}_{ij}\}
%\text{if $\ket{\phi}$ is the unique entangled 2-qubit state satisfying $\Pi_{e}$}
\end{cases}
$

\EndIf
%\vspace{.5em} 
\EndFor
\vspace{.5em} 
%\LState
\If {$s_0$ is not coherent for some $i \in V_s$ or for three
distinct variables $i,j,k$, we have $s(i) = \ket{\gamma}_{ik}$ and 
$\Pi_e$ is an entangled projector on $i$ and $j$}
%$s$ doesn't satisfy $\Pi_e$ for some edge $e$ from $i$ to $j$ with $s(i) \neq \Empty \neq s(j)$
\State
{\bf output} ``$H$ is unsatisfiable"
\EndIf
\vspace{.5em} 
\State
{\bf remove} from $E_0$ every edge $e$ such that $\Pi_e$ is satisfied by  $s_0.$
\State
{\bf remove} every isolated vertex from $G_0$ 

%\LState
%{\bf while} {$s$ is coherent and not closed} {\bf do}
\State
$s_1 := s_0, ~~G_1 := G_0$
\end{algorithmic}
\end{algorithm}

\begin{Lem}
\label{lem:maxrank} 
  Let $s_0, G_0, s_1, G_1$ be the outcome of {\sf
  MaxRankRemoval}. Then:
  \begin{enumerate}
    \item If {\sf MaxRankRemoval} doesn't output 
      ``$H$ is unsatisfiable" then $s_0$ is coherent, it satisfies
      every maximal rank constraint, $G_0 = G(H_{s_0})$ and
      $s_0=s_1, G_0=G_1$. Moreover, if $H$ is satisfiable then $s_0$
      is a pre-solution.
      
    \item If {\sf MaxRankRemoval} outputs 
      ``$H$ is unsatisfiable" then indeed $H$ is unsatisfiable.
      
    \item The complexity of the procedure is $O(|V| + |E|) |$.
  \end{enumerate}
\end{Lem}
\begin{proof}
If the procedure doesn't output ``$H$ is unsatisfiable" then indeed
$s_0$ is coherent and it satisfies all maximal rank constraints. The
removal of the necessary edges and vertices insures that $G_0 =
G(H_{s_0})$, and obviously $s_0=s_1, G_0=G_1$. If $H$ is
satisfiable, then it has a ground state for some total assignment
$s$. This $s$ is an extension of $s_0$ because there is a unique way
to satisfy the maximal rank constraints.

Maximal rank projectors are such that there is a unique assignment
for their qubit(s) which satisfies them. The first part of the
procedure creates the assignment which assigns these necessary
values. If this assignments is not coherent then $H$ is
unsatisfiable. 
Similarly, if $s_0$ assigns an entangled 2-qubit state between variables $i$ and $k$, and 
there is an entangled rank-1 constraint between $i$ and  $j$, then by Lemma~\ref{lem:propagate2}
it is impossible to extend $s_0$ into a satisfying assignment,
and therefore $H$ is unsatisfiable. This proves
Statement 2.
%if $s(i)$ and $s(j)$ get assigned but the
%constraint $\Pi_e$ for and edge from $i$ to $j$ is unsatisfied, then
%it is impossible to satisfy simultaneously $H_e$ and the maximal
%rank constraints, and therefore $H$ is unsatisfiable. This proves Statement 2.

The procedure can be executed by a constant number of vertex and edge traversals for $s_0$, and
similarly for $s_1$.
\end{proof}

%XXXXXXXXXXXXXXXXXXXXXXXXXXXXX parallel-prop_short XXXXXXXXXXXXXXXXXXXXXXXXXX

% token holder: no one

%====================================================================
\subsection{Algorithm {\sf ParallelPropagation}}\label{sec:Prop2}

%\begin{comment}		
The procedure {\sf ParallelPropagation} is called when $s_0$ is a
closed assignment, and in $G_{s_0}$ there is a product edge. Since
there are only two ways to satisfy a product constraint, these are
tried and propagated in parallel. If one of these propagations
terminates successfully, the other is stopped, which ensures that
the overall work done is proportional to the progress made.
%\mnote{Same problem of macro here.} \mnote{Corrected}

\begin{algorithm}
\caption{{\sf ParallelPropagation}$(i_0, \ket{\alpha_0}, i_{1}, \ket{\alpha_{1}})$}
\label{alg:Q2SATdotted}
\begin{algorithmic}[H]
%        \LState $s(i_0) := \{\ket{\alpha_0}, ~~~ s(i_1) := \{\ket{\alpha_1}\}$
	%\LState {\bf Run} in parallel  {\sf Propagation}$(s_0, i_0,\ket{\alpha_0})$ and 
	 %{\sf Propagation}$(s_1, i_1,\ket{\alpha_1})$ (propagation 0 and 1)
		%\LState (Call them propagation 0 and propagation 1, respectively.)
		%\State
	%$s_0(i_0) := \ket{\alpha_0^\bot}, s_1(i_1) := \ket{\alpha_1^\bot}$
	\State  {\bf Run} in parallel  {\sf Propagation}$(s_0, G_0, i_0, \ket{\alpha_0})$ %(called propagation $b$) 
	and {\sf Propagation}$(s_1, G_1, i_1, \ket{\alpha_1})$ 
	%for $b \in \{0,1\}$
	%{\sf Propagation}$(s', i',\ket{\alpha'})$
	%(called propagation $b$) for $ b=0,1$
	\State {\bf until} one of them terminates successfully {\bf or} both terminate unsuccessfully
	
	\medskip
	%\If {propagation $i$ stops first and $f_i = 1$} 
		%\LState Abandon propagation $i$ .
		%\LState Let the other propagation $(1-i)$ continue until it cannot propagate. 
		%\LState $j = 1-i$.
	%\ElsIf {propagation $i$ stops first and $f_i = 0$} 
		%\LState End the other propagation $(1-i)$ (if it has not stopped).
		%\LState $j = i$.
	%\EndIf
         \If {both propagations terminate unsuccessfully}
         \State {\bf output} ``$H$ is unsatisfiable'' %and terminate the whole program
	%\Else 
		%\LState $s_0(u_j) = \ket{\alpha_j}$.
		%\LState {\sf Propagation}$(u_j,\ket{\alpha_j},1)$ (with elements written in $s_0$).
	\Else ~~~let {\sf Propagation}$(s_0, G_0, i_0, \ket{\alpha_0})$ terminate first (the other case is symmetric)
	\State {\bf undo} {\sf Propagation}$(s_1, G_1, i_1, \ket{\alpha_1})$
	%{\sf Propagation}$(s_{1}, G_1,i_0,\ket{\alpha_0})$
	\State
	$s_1 := s_0, ~G_1 :=G_0$
	%\LState {\sf Propagation}$(s', i,\ket{\alpha})$
	\EndIf	
	
\end{algorithmic}
\end{algorithm}

\begin{Lem}
\label{lem:parallel} 

  Let {\sf ParallelPropagation} be called when $s_0$ is closed,
  $H_{s_0}$ doesn't have rank-$3$ constraints, $G_0=G_{s_0}$, in
  $G_0$ there exists a product edge from $i_0$ to $i_1$ with
  constraint $\ket{\alpha_0^\bot}_{}\otimes \ket{\alpha_1^\bot}_{}$, and $s_1
  = s_0, G_1 = G_0$.  Let $s'_0, s'_1, G'_0 , G'_1$ be the outcome
  of the procedure. Then:
  \begin{enumerate}
    \item If {\sf ParallelPropagation} doesn't output 
      ``$H$ is unsatisfiable'' then $s'_0$ is a proper closed
      extension of $s_0$, $G'_0 = G_{s'_0}$, and $s'_1=s'_0$,
      $G'_1=G'_0$. Moreover, if $s$ is a pre-solution then $s'_0$ is
      a pre-solution.
    
    \item If {\sf ParallelPropagation} outputs 
      ``$H$ is unsatisfiable'' then indeed $H$ is unsatisfiable.
    
    \item The complexity of the procedure is 
      $O( |E_{s_0} |  -  |E_{s'_0}|)$.
  \end{enumerate}
\end{Lem}

\begin{proof}
  If the procedure doesn't output ``$H$ is unsatisfiable" then at
  least one of the parallel propagations terminates successfully,
  say {\sf Propagation}$(s_0, G_0, i_0, \ket{\alpha_0})$. Then $s'_0$ is a proper
  extension of $s_0$ since $s_0$ is closed and therefore $s_0(i_0) =
  \Empty$. Obviously $s'_1=s'_0$ and $ G'_1=G'_0$, and all other
  claims follow from the Propagation Lemma.

  Since $H_{s_0}$ doesn't have rank-3 constraints, by
  Theorem~\ref{thm:productstate} if it is frustration free, it has
  a product ground state. In $H_{s_0}$ there exists a product
  edge from $i_0$ to $i_1$ with constraint $\ket{\alpha_0^\bot}_{i_0}
  \otimes \ket{\alpha_1^\bot}_{i_1} \bra{\alpha_0^\bot}_{i_0} \otimes
  \bra{\alpha_1^\bot}_{i_1} $, therefore only assignments which
 have either $\ket{\alpha_0}$ assigned to variable $i_0$ or
  $\ket{\alpha_1}$ assigned to variable $i_1$ can be a solution. But if both
  propagations output ``unsuccessful", then by the Propagation Lemma
  no such assignment can satisfy $H_{s_0}$. Therefore
  $H_{s_0}$ is not frustration free, and neither is $H$.

  For the complexity analysis observe that the unsuccessful or
  unterminated propagation of the parallel processes makes at most
  as many steps as the successful one. This is the reason for 
  performing the two propagations in parallel. Undoing this
  propagation can be performed in the same order of time as the
  propagation itself, for example, by copying the removed edges into
  temporary lists. The claim on the complexity of the successful
  propagation follows from the Propagation Lemma.
\end{proof}

%XXXXXXXXXXXXXXXXXXXXXXXXXXXXX probe-prop_short XXXXXXXXXXXXXXXXXXXXXXXXXX
% token holder: No one

%====================================================================
\subsection{Algorithm {\sf ProbePropagation}}\label{sec:Prop1}

The procedure {\sf ProbePropagation} is evoked when $s_0$ is a
closed assignment, and in $G_{s_0}$ there are only entangled
constraints. It picks an arbitrary vertex in $i \in V_s$, assigns
$\ket{0}$ (an arbitrary value) to it, and propagates this choice. In
the lucky case of successful propagation this is repeated.
Otherwise, we reach a contradiction: there is some $j \in V_s$,
such that two propagating paths assign different values to it. We
prove below the Sliding Lemma which already appeared in \Ref{JWZ11}.
It implies that when $ i_0\to i_1\to \ldots \to i_k $ is a propagating
path of entangled rank-1 projectors, the ground space of the
Hamiltonian $\Pi_{i_0,i_1}+\Pi_{i_1,i_2} + \ldots +
\Pi_{i_{k-1},i_k}$ is equal to the ground state of the Hamiltonian
$\Pi_{i_0,i_k} + \Pi_{i_1,i_2} + \ldots + \Pi_{i_{k-1},i_k}$, where
$\Pi_{i_0,i_k}$ is a new projector defined on the qubits $(i_0,i_k)$
that replaces the projector $\Pi_{i_0,i_1}$. Graphically, this can
be viewed as if we are sliding the $(i_0,i_1)$ edge on the path
$i_1\to \ldots \to i_k$, as shown in Fig.~\ref{fig:sliding-path}.
Therefore, if we have two propagating paths starting at $i$ and
ending at $j$, they define two projectors on qubits $(i,j)$, as
illustrated in Fig.~\ref{fig:contradicting-cycle}. As we
shall see, if these two paths are contradicting then necessarily the
two projectors are different, which by Lemma~\ref{lem:rank2product}
implies the existence of a product constraint between $(i_0,i_k)$
variables. In such case, we can proceed by calling the procedure
{\sf ProbePropagation}.

\begin{figure}
  \begin{center}
    \includegraphics[scale=1.0]{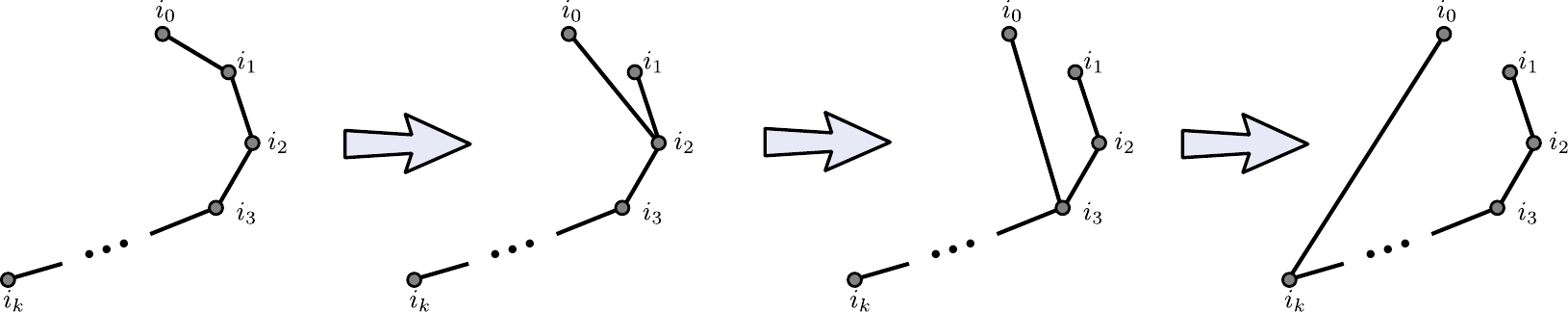}  
  \end{center}
  \caption{\label{fig:sliding-path} The sliding of the edge
  $(i_0,i_1)$ over the path $i_1\to i_2\to \ldots \to i_k$, until it
  becomes the edge $(i_0,i_k)$.}
\end{figure}

\begin{Lem}[Sliding Lemma]
\label{lem:sliding} Consider a system on $3$ qubits $i,j$ and $k$.
  Suppose that we have a two rank-$1$ constraints $\Pi_1 =
  \ket{\psi_1}\bra{\psi_1}_{ij}$ on qubits $(i,j)$ and $\Pi_2 =
  \ket{\psi_2}\bra{\psi_2}_{jk}$ on qubits $(j,k)$. If
  $\ket{\psi_2}$ is entangled, there is another rank-$1$
  constraint $\Pi_3 = \ket{\psi_3}\bra{\psi_3}_{ik}$ on qubits
  $(i,k)$ such that the ground space of $\Pi_1+\Pi_2$ 
  is identical to the ground space of $\Pi_2+\Pi_3$.
  In addition, if a single qubit
  state $\ket{\alpha}_i$ is propagated by $\Pi_1+\Pi_2$ to
  $\ket{\beta}_k$, then it is also propagated to $\ket{\beta}_k$
  directly via $\Pi_3$.
\end{Lem}
\begin{proof}
  Consider the Schmidt decomposition 
  $\ket{\psi_2}_{jk} = \lambda_1 \ket{x_1}_j\ket{y_1}_k +
  \lambda_2\ket{x_2}_j\ket{y_2}_k$, where 
  $\lambda_1,\lambda_2\neq 0$, as $\ket{\psi_2}_{jk}$ is entangled.
  Define a non-singular transformation $T$ on qubit $j$ by
  $\lambda_1 T\ket{x_1} = \ket{y_2}$ and $\lambda_2T\ket{x_2} =
  -\ket{y_1}$. Then $T\ket{\psi_2}_{jk} =
  \ket{y_2}_j\ket{y_1}_k - \ket{y_1}_j\ket{y_2}_k$ is the
  anti-symmetric state. Let $\ket{\tilde{\psi}_1}_{ij}$ and 
  $\ket{\tilde{\psi}_2}_{jk}$ be the normalization of
  $T\ket{\psi_1}_{ij}$ and $T\ket{\psi_2}_{jk}$ respectively, and
  use them to define the rank-$1$ projectors
  $\tilde{\Pi}_1,\tilde{\Pi}_2$. Since $\tilde{\Pi}_2$ projects into
  the anti-symmetric subspace, then any state in the ground space of
  $\tilde{\Pi}_1+\tilde{\Pi}_2$ must be invariant under a swapping of
  qubits $j,k$. Therefore, definining
 $\ket{\psi_3}_{ik}=\ket{\tilde{\psi}_1}_{ik}$, and
  $\Pi_3=\Id-\ket{\psi_3}_{ik} \bra{\psi_3}_{ik}$, the the ground
  space of $\tilde{\Pi}_1+\tilde{\Pi}_2$ is identical to the ground
  space of $\Pi_3+\tilde{\Pi}_2$. Applying now the inverse
  transformation $T^{-1}$ on qubit $j$, the projector
  $\tilde{\Pi}_2$ returns to $\Pi_2$, while $\Pi_3$ remains
  unchanged. Since both $T,T^{-1}$ are non-singular, it
  follows that ground space of $\Pi_1+\Pi_2$ is identical to the
  ground space to $\Pi_2+\Pi_3$. 
  
  For the second claim, assume by contradiction that $\Pi_3$ does
  not propagate $\ket{\alpha}_i$ to $\ket{\beta}_k$. Then there is a
  1-qubit state $\ket{\gamma}\neq \ket{\beta}$, such that
  $\Pi_3(\ket{\alpha}_i\ket{\gamma}_k)=0$. Since $\Pi_2$ is a rank-1
  entangled projector, by Lemma~\ref{lem:propagate0}, it propagates
  $\ket{\gamma}_k$ to some state $\ket{\delta}_j$, and therefore the
  state $\ket{\alpha}_i\ket{\delta}_j\ket{\gamma}_k$ is a ground
  state of $\Pi_2+\Pi_3$, as well as of $\Pi_1+\Pi_2$. But this
  contradicts the assumption that latter propagates 
  $\ket{\alpha}_i$ to $\ket{\beta}_k$.  
\end{proof}

Applying Lemma~\ref{lem:sliding} iteratively, we reach the following
corollary
\begin{Cor}
\label{cor:sliding} Let $H = \sum_{e\in I} H_e$ be a $2$-local
  Hamiltonian in rank-$1$ decomposition.  Let $ i_0, i_1, \ldots
  i_k$ be vertices in $V$, and let $e_j$ be an edge from $i_j$ to
  $i_{j+1}$, for $j = 0, \ldots, k-1$ such that the rank-$1$
  constraints $\Pi_{e_j}$ are entangled. Then there exists a
  $2$-qubit entangled state $\ket{\gamma}$ between $i_0$ and $i_k$
  such that the ground space of $\sum_{j=0}^{k-1} \Pi_{e_j}$ is
  identical to the ground space of $\sum_{j=1}^{k-1} \Pi_{e_j} +
  \ket{\gamma}\bra{\gamma}_{i_0,i_k}$.  Moreover, if
  $\ket{\alpha}_{i_0}$ is propagated to $\ket{\beta}_{i_k}$ along
  the path, then it is also propagated directly by
  $\ket{\gamma}\bra{\gamma}_{i_0,i_k}$.
\end{Cor}

We will denote the state $\ket{\gamma}$ in the conclusion of the
corollary by $\slide(p)$.

\begin{algorithm}
\caption{{\sf ProbePropagation}$(i)$}
\label{alg:prop1}
\begin{algorithmic}[H]
%\State
%$s(i) = \ket{0}$
	\State {\sf Propagation}$(s_0, G_0, i, \ket{0})$. %Record the BFS tree $T$. 

	\If {the propagation is successful}
		%\LState {\sf Propagation}$(s_1, G_1,i,\ket{0}$
		$s_1 := s_0$, $G_1 := G_0$
		%\LState \Return  $(0,1,\ket{0},\ket{0})$,
	\Else %\Comment{a contradiction found}
	\State Let $j$ such that $|s_0(j)| > 1$
	\State {\bf find} two paths $p_1$ and $p_2$ in $G_0$ from $i$
	to $j$ such that $\pro(s_0, p_1, \ket{0}) \neq \pro(s_0, p_2, \ket{0})$
		
	%\medskip
	\State {\bf find} a product state $\ket{\alpha^\bot}_i \otimes \ket{\beta^\bot}_j$ in the two dimensional space
	$\spa\{\slide(p_1), \slide(p_2)\}$
		%\LState Slide the constraint along $p_1$ to find a rank-1 projector $\ket{\psi_1}\bra{\psi_1}$ on $(k,j)$.
		%\LState Slide the constraint along $p_2$ to find a rank-1 projector $\ket{\psi_2}\bra{\psi_2}$ on $(k,j)$.
		%\LState Let $\Pi$ be the rank-2 projector on the subspace $span\{\ket{\psi_1},\ket{\psi_2}\}$ in qubits $(k,j)$. 
		%\LState Use Lemma \ref{lem:rank2product} on $\Pi$ to find a product state $\ket{\alpha}_k\otimes \ket{\beta}_j$.
		\State {\bf undo} {\sf Propagation}$(s_0, G_0,i, \ket{0})$
		\State {\sf ParallelPropagation$(i,\ket{\alpha}, j,\ket{\beta})$}
	\EndIf
\end{algorithmic}
\end{algorithm}

\begin{Lem}
\label{lem:probe} 
  Let {\sf ProbePropagation} be called when $s_0 $ is closed,
  $H_{s_0}$ has only rank-$1$ entangled constraints,
  $G_0=G_{s_0}$, and $s_1 = s_0, G_1 = G_0$.  Let $s'_0, s'_1,
  G'_0 , G'_1$ be the outcome of the procedure. Then:
  \begin{enumerate}
    \item If {\sf ProbePropagation} doesn't output 
    ``$H$ is unsatisfiable'' then $s'_0$ is a proper closed
    extension of $s_0$, $G'_0 = G_{s'_0}$, and $s'_1=s'_0$,
    $G'_1=G'_0$. Moreover, if $s$ is a pre-solution then $s'_0$ is a
    pre-solution.
    
    \item If {\sf ParallelPropagation} outputs 
      ``$H$ is unsatisfiable'' then indeed $H$ is unsatisfiable.
    
    \item The complexity of the procedure is $O( |E_{s_0} |  -  |E_{s'_0}|)$.
  \end{enumerate}
\end{Lem}

\begin{proof}
  If the procedure doesn't output ``$H$ is unsatisfiable" then 
  either {\sf Propagation}$(s_0, G_0, i, \ket{0})$ or one of the parallel
  propagations (say {\sf Propagation}$(s_0, G_0, i, \ket{\alpha})$) terminates
  successfully. Then $s'_0$ is a proper extension of $s_0$ since
  $s_0$ is closed and therefore $s_0(i_0) = \Empty$. Obviously
  $s'_1=s'_0$ and $ G'_1=G'_0$, and all other claims follow from the
  Propagation Lemma.

  Let's suppose that all three propagations are unsuccessful. By
  Corollary~\ref{cor:sliding}, any solution for $H_{s_0}$ also
  satisfies $\ket{\alpha^\bot}_i \otimes \ket{\beta^\bot}_j \bra{\alpha^\bot}_i
  \otimes \bra{\beta^\bot}_j $.  Then Lemma~\ref{lem:parallel} implies
  that $H$ is unsatisfiable.

  For the complexity analysis the interesting case is when the first
  propagation, that we call \texttt{Prop$_{{\rm failure}}$}, is unsuccessful but one of
  the two parallel propagations is successful. Let's call this
  successful one \texttt{Prop$_{{\rm success}}$}.  The main observation here is that
  every propagating edge in \texttt{Prop$_{{\rm failure}}$} will also be propagating in
  \texttt{Prop$_{{\rm success}}$}, since by Lemma~\ref{lem:propagate0} entangled edges
  always propagate. The paths $p_1$ and $p_2$ can be found in time
  proportional to the size of the subgraph visited by \texttt{Prop$_{{\rm failure}}$}.
  Indeed, observe that the edges of the two paths, except the last edge of one of the two,
  are edges in the BFS tree underlying \texttt{Prop$_{{\rm failure}}$}.
  The way from a vertex to the root of the tree can be then found, for example, by maintaining for each vertex in the tree,
  a pointer towrads its father.
  The product state $\ket{\alpha} \otimes \ket{\beta}$ can be
  found in constant time by Lemma~\ref{lem:rank2product}. Therefore,
  by Lemma~\ref{lem:parallel}, the complexity is indeed $O( |E_{s_0}
  | - |E_{s'_0}|)$.
\end{proof}

\subsection{Analysis of the algorithm}
\label{subsec:analysis}

\textbf{Proof of Theorem~\ref{thm:main}} 
If $H$ is frustration free then by Lemma~\ref{lem:maxrank} {\sf
MaxRankRemoval} outputs a pre-solution $s_0$ that satisfies every
maximal rank constraint. By the Propagation Lemma, at the end of
Phase two, in addition $s_0$ is a closed. By
Lemma~\ref{lem:parallel} {\sf ParallelPropagation} outputs $s_0$
such that in addition in $H_s$ there are only entangled constraints.
By Lemma~\ref{lem:probe} at the end of the algorithm in addition
$H_s$ is empty, and therefore $s$ is a solution.

If the algorithm doesn't output ``$H$ is unsatisfiable" then
by  Lemma~\ref{lem:maxrank}, by the Propagation Lemma, and by Lemmas~\ref{lem:parallel}
and~\ref{lem:probe} it outputs 
a coherent assignment $s$ such that $G_s$ is the empty graph, and therefore $s$ 
is a solution.

%The algorithm either outputs $H$ is ``unsatisfiable'' or a ground
%state. If $H$ is unsatisfiable, then obviously it can not output a ground state.

The complexity of {\sf MaxRankRemoval} by Lemma~\ref{lem:maxrank} 
is $O(|E|)$. After the second phase, the propagation of the assigned values 
during {\sf MaxRankRemoval}, the copying of $s_0$ and $G_0$ into respectively
$s_1$ and $G_1$ can be done by executing the same propagation steps this time
with $s_1$ and $G_1$. The complexity of the rest of the algorithm by the 
Propagation Lemma, and Lemmas~\ref{lem:parallel} and \ref{lem:probe}
is a telescopic sum which sums up to also $O(|E|)$.

\bibliographystyle{alpha}
\bibliography{qhc}

\end{document}